\documentclass[10]{llncs}
\usepackage{float}
\floatstyle{ruled}
\usepackage{boxedminipage}
\usepackage{graphicx, subfigure}
\usepackage{epsfig}
\usepackage{theorem}
\usepackage{latexsym}
\usepackage{amssymb}
\newfloat{algorithm}{tbp}{loa}
\floatname{algorithm}{Algorithm}
\usepackage{setspace}

\renewenvironment{proof}{{\it Proof. } }{{\hfill $\Box$}\vspace{.5pc}}

\newtheorem{specification}{Specification}
\newtheorem{lem}{Lemma}
\newcommand{\keywords}[1]{\par\addvspace\baselineskip
\noindent\keywordname\enspace\ignorespaces#1}

\newcommand{\ie}{\emph{i.e., }}

\begin{document}

\title{Snap-Stabilizing Message Forwarding algorithm on tree topologies %
        \thanks{This work is supported by ANR SPADES grant.}
}

\author{
  Alain Cournier\inst{1} \and  
  Swan Dubois\inst{2} \and 
  Anissa Lamani\inst{1} \and 
  Franck Petit\inst{2} \and 
  Vincent Villain\inst{1} 
}

\authorrunning{Cournier {\em et al}.}   
%
\tocauthor{ Alain Cournier, Swan Dubois, Anissa Lamani, Franck Petit, Vincent Villain}
\institute{
MIS, Universit\'e of Picardie Jules Verne, France 
\and
LiP6/CNRS/INRIA-REGAL, Universit\'e Pierre et Marie Curie - Paris 6, France 
}

\maketitle

\begin {abstract}

In this paper, we consider the message forwarding problem that consists in managing the network resources that are used to forward messages. Previous works on this problem provide solutions that either use a significant number of buffers (that is $n$ buffers per processor, where $n$ is the number of processors in the network) making the solution not scalable or, they reserve all the buffers from the sender to the receiver to forward only one message 
. The only solution that uses a constant number of buffers per link was introduced in \cite{CDLPV10}. However the solution works only on a chain networks. In this paper, we propose a snap-stabilizing algorithm for the message forwarding problem that uses the same complexity on the number of buffers as \cite{CDLPV10} and works on tree topologies.

\end {abstract}

\section{Introduction}
It is known that the quality of a distributed system depends on its fault tolerance. Many fault-tolerance approaches have been introduced, for instance: Self-Stabilization \cite{D00} which allows the conception of systems that are tolerant of any arbitrary transient fault. A system is said to be self-stabilizing if starting from any arbitrary configuration, the system converges into the intended behavior in a finite time. Another instance of the fault-tolerance scheme is the snap-stabilization \cite{Bui07}. Snap-stabilizing systems always behave according to their specification, and this regardless of the starting configuration. Thus a snap-stabilizing solution can be seen as a self-stabilizing solution that stabilizes in zero time. 

In distributed systems, the {\em end-to-end communication} problem consists in delivery in finite time across the network of a sequence of data items generated at a node called the sender, to another node called the receiver. This problem comprises the following two sub-problems:
($i$) the {\em routing} problem, \ie the determination of the path followed by the messages to reach their destinations;
($ii$) the {\em message forwarding} problem that consists in the management of network resources in order to forward messages. In this paper, we focus on the second problem whose aim is to design a protocol that manages the mechanism allowing the message to move from a node to another one on the path from a sender to a receiver. Each node on this path has a reserved memory space called buffer. With a finite number of buffers, the message forwarding problem consists in avoiding deadlock and livelock situations.

The message forwarding problem has been well investigated in a non faulty setting~\cite
{Duato96,MerlinS78,Toueg80,TouegU81}. In \cite{APV96,KOR95} self-stabilizing solutions were proposed.  Both solutions deal with
network dynamic, \ie systems in which links can be added or removed.  However, they 
assume that the routing tables are correct (loop-free). Thus the proposed
solutions cannot ensure absence of deadlocks or message loss during the stabilization time. 

In this paper, we address the problem of providing a snap-stabilizing protocol for this problem. 
Snap-stabilization provides the desirable property of delivering to its recipient every message generated after
the faults, once and only once  even if the routing tables are not (yet) stabilized. 
Some snap-stabilizing
solutions have been proposed to solve the problem~\cite{CDV09-1,CDV09-2,CDLPV10}. In \cite{CDV09-1}, the problem
was solved using $n$ buffers per node (where $n$ denotes the number of processors in the network). The number of
buffers was reduced in \cite{CDV09-2} to $D$ buffers per node (where $D$ refers to the diameter of the network).
However, the solution works by reserving the entire sequence of buffers leading from the sender to the receiver. Note
that the first solution is not suitable for large-scale systems whereas the second one has to reserve all the path
from the source to the destination for the transmission of only one message.
In \cite{CDLPV10}, a snap-stabilizing solution was proposed using a constant number of buffers per link. However the solution works only on chain topologies.

We provide a snap-stabilizing solution that solves the message forwarding problem in tree topologies
using the same complexity on the number of buffers as in \cite{CDLPV10} \ie $2\delta+1$ buffers by processor, where
$\delta$ is the degree of the processor in the system. 

\paragraph{\textbf{Road Map}}

The rest of the paper is organized as follow: Our Model is presented in Section \ref{sec:Model}. In Section \ref{sec:Solution}, we provide our snap-stabilizing solution for the message forwarding problem. The proofs of correctness are given in Sub-Section \ref{subsec:Proofs}. Finally we conclude the paper in Section \ref{sec:Conclusion}.

\section{Model and Definitions}\label{sec:Model}

\paragraph{\textbf{Network}.} 
We consider in this paper a network as an undirected connected graph $G=(V,E)$ where $V$ is the set of nodes (processors) and $E$ is the set of bidirectional communication links. Two processors $p$ and $q$ are said to be neighbours if and only if there is a communication link $(p,q)$ between the two processors. Note that, every processor is able to distinguish all its links. To simplify the presentation we refer to the link $(p,q)$ by the label $q$ in the code of $p$. In our case we consider that the network is a tree of $n$ processors.\\

\paragraph{\textbf{Computational model}.}
In this paper we consider the classical local shared memory model introduced by Dijkstra~\cite{D74} known as the state model. In this model communications between neighbours are modelled by direct reading of variables instead of exchange of messages. The program of every processor consists in a set of shared variables (henceforth referred to as variable) and a finite number of actions. Each processor can write in its own variables and read its own variables and those of its neighbours. Each action is constituted as follow:

\begin{center} $<Label>::<Guard>$ $\rightarrow$ $<Statement>$ \end{center}

The guard of an action is a boolean expression involving the variables of $p$ and its neighbours. The statement is an action which updates one or more variables of $p$. Note that an action can be executed only if its guard is true. Each execution is decomposed into steps. 

The state of a processor is defined by the value of its variables. The state of a system is the product of the states of all processors. The local state refers to the state of a processor and the global state to the state of the system.  

Let $y$ $\in$ $C$ and $A$ an action of $p$ ($p$ $\in$ $V$). $A$ is {\em enabled} for $p$ in $y$ if and only if the guard of
$A$ is satisfied by $p$ in $y$. Processor $p$ is enabled in $y$ if and only if at least one action is enabled at $p$
in $y$. Let $P$ be a distributed protocol which is a collection of binary transition relations denoted by
$\rightarrow$, on $C$. An execution of a protocol $P$ is a maximal sequence of configurations  $e=
y_{0}y_{1}...y_{i}y_{i+1} \ldots$ such that, $\forall$ $i\ge0$, $y_{i} \rightarrow y_{i+1}$ (called a step) if $y_{i+1}$
exists, else $y_{i}$ is a terminal configuration. {\em Maximality} means that the sequence is either finite (and no action
of $P$ is enabled in the terminal configuration) or infinite. All executions considered here are assumed to be
maximal. $\xi$ is the set of all executions of $P$. 
Each step consists on two sequential phases atomically executed:
($i$) Every processor evaluates its guard;
($ii$) One or more enabled processors execute its enabled actions. 
When the two phases are done, the next step begins. 
This execution model is known as the \emph{distributed daemon}~\cite{BGM89}. 
We assume that the daemon is \emph{weakly fair}, meaning that if a processor $p$ is continuously $enabled$, 
then $p$ will be eventually chosen by the daemon to execute an action.

In this paper, we use a composition of protocols.  We assume that the above statement ($ii$) is applicable to every
protocol. In other words, each time an enabled processor $p$ is selected by the daemon, $p$ executes the enabled actions of every protocol.

\paragraph{\textbf{Snap-Stabilization}.} 
Let $\Gamma$ be a task, and $S_{\Gamma}$ a specification of $\Gamma$ . A protocol $P$ is snap-stabilizing for $S_{\Gamma}$ if and only if $\forall \Gamma \in \xi$, $\Gamma$ satisfies $S_{\Gamma}$.


\paragraph{\textbf{Message Forwarding Problem}.}


The message forwarding problem is specified as follows: 

\begin{specification}[$SP$]\label{spec:SP}

A protocol $P$ satisfies $SP$ if and only if the following two requirements are satisfied in every execution of $P$:
$(i)$ Any message can be generated in a finite time.
$(ii)$ Any valid message is delivered to its destination once and only once in a finite time.
\end{specification}

\paragraph{\textbf{Buffer Graph}} 

A Buffer Graph \cite{MS78} is defined as a directed graph on the buffers of the graph \ie the nodes are a subset of the buffers of the network and links are arcs connecting some pairs of buffers, indicating permitted message flow from one buffer to another one. Arcs are only permitted between buffers in the same node, or between buffers in distinct nodes which are connected by a communication link. 



\section{Message Forwarding}\label{sec:Solution}

In this section, we first give an overview of our snap stabilizing Solution for the message forwarding problem, then we present the formal description followed by some sketches of the proofs of correctness.

\subsection{Overview of the Solution}\label{subsec:Overview}

In this section, we provide an informal description of our snap stabilizing solution that solves the message forwarding problem and tolerates the corruption of the routing tables in the initial configuration.
We assume that there is a self-stabilizing algorithm that calculates the routing tables and runs simultaneously to our algorithm. We assume that our algorithm has access to the routing tables via the function $Next_p(d)$ which returns the identity of the neighbour to which $p$ must forward the message to reach the destination $d$. In the following we assume that there is no message in the system whose destination is not in the system.\\

 
Before detailing our solution let us define the buffer graph used in our solution: 


Let $\delta(p)$ be the degree of the processor $p$ in the tree structure. Each processor $p$ has $(i)$ one internal buffer that we call Extra buffer denoted $EXT_p$. $(ii)$ $\delta(p)$ input buffers allowing $p$ to receive messages from its neighbors. Let $q\in N_p$, the input buffer of $p$ connected to the link $(p,q)$ is denoted by $IN_{p}(q)$. $(iii)$ $\delta(p)$ output buffers allowing it to send messages to its neighbors.  Let $q\in N_p$, the output buffer of $p$ connected to the link $(p,q)$ is denoted by $OUT_{p}(q)$. In other words, each processor $p$ has $2\delta(p)+1$ buffers. The generation of a message is always done in the output buffer of the link $(p,q)$ so that, according to the routing tables, $q$ is the next processor for the message in order to reach its destination.  

The overall idea of the algorithm is the following: When a processor wants to generate a message, it consults the routing tables to determine the next neighbour by which the message will transit in order to reach its destination. Once the message is on system, it is routed according to the routing tables: Let us refer to $nb(m,b)$ as the next buffer $b'$ of the message~$m$ stored in $b$,  
$b \in \{IN_{p}(q) \vee OUT_{p}(q)\}$, $q\in N_{p}$. We have the following properties:
\begin{enumerate}
 \item $nb(m,IN_{p}(q))= OUT_{p}(q')$ such as $q'$ is the next process by which $m$ has to transit to reach its destination.
\item $nb(m,OUT_{p}(q))= IN_{q}(p)$
 \end{enumerate}
 
Thus, if the message $m$ is in the Output buffer $OUT_p(q)$ such as $p$ is not the destination then it will be automatically copied in the Input buffer of $q$. If the the message $m$ is in the Input buffer of $p$ ($IN_p(q)$) then if $p$ is not the destination it consults the routing tables to determine which is the next process by which the message has to pass in order to meet its destination. 

Note that when the routing tables are stabilized and when all the messages are in the right direction, the first property $nb(m,IN_{p}(q))= OUT_{p}(q')$ is never verified for $q=q'$. However, this is not true when the routing tables are not yet stabilized and when some messages are in the wrong direction.
%

Let us now recall the message progression. A buffer is said to be free if and only if it is empty (it contains no message) or contains the same message as the input buffer before it in the buffer graph buffer. In the opposite case, a buffer is said to busy. The transmission of messages produces the filling and the cleaning of each buffer, i.e., each buffer is alternatively free and busy. This mechanism clearly induces that free slots move into the buffer graph, a free slot corresponding to a free buffer at a given instant. 
%
%
%

In the following, let us consider our buffer graph taking in account only active arcs (an arc is said to be active if it starts from a non empty buffer). Observe that in this case the sub graph introduced by the active arcs can be seen as a resource allocation graph where the buffers correspond to the resources, for instance if there is a message $m$ in $IN_p(q)$ such as $nb(m,IN_{p}(q))= OUT_{q'}(p)$ then $m$ is using the resource (buffer) $IN_p(q)$ and it is asking for another resource which is the output buffer $OUT_p(q')$. In the following we will refer to this sub graph as the active buffer graph. 

It is known in the literature that a deadlock situation appears only in the case there exists a cycle in the resource allocation graph. Note that this is also the case in our active buffer graph. Observe that because our buffer graph is built on a tree topology, if a cycle exists then we are sure that there are at least two messages $m$ and $m'$ that verifies the following condition: $nb(m,IN_{p}(q))= OUT_{p}(q)$ $\wedge$ $nb(m',IN_{p'}(q'))= OUT_{p'}(q')$. Since in this paper we consider a distributed system, it is impossible for a processor $p$ to know whether there is a cycle in the system or not if no mechanism is used to detect them. The only thing it can do is to suspect the presence of a cycle in the case there is one message in its input buffer $IN_p(q)$ that has to be sent to $OUT_p(q)$. In order to verify that, $p$ will initiate a token circulation that will follow the active buffer graph starting from the input buffer containing the message $m$. By doing so, the token circulation either finds a free buffer (refer to Figure \ref{Tokenn}, (b)) or detects a cycle. Note that two kinds of cycle can be detected: $(i)$ a Full-Cycle involving the first input buffer containing $m$ (refer to Figure \ref{Tokenn}, (a)) or $(ii)$ a Sub-Cycle that does not involve the input buffer that contains the message $m$ (refer to Figure \ref{Tokenn}, (c)). 

\begin{figure}
 \begin{minipage}[b]{.46\linewidth}
  \begin{center}
  \epsfig{figure=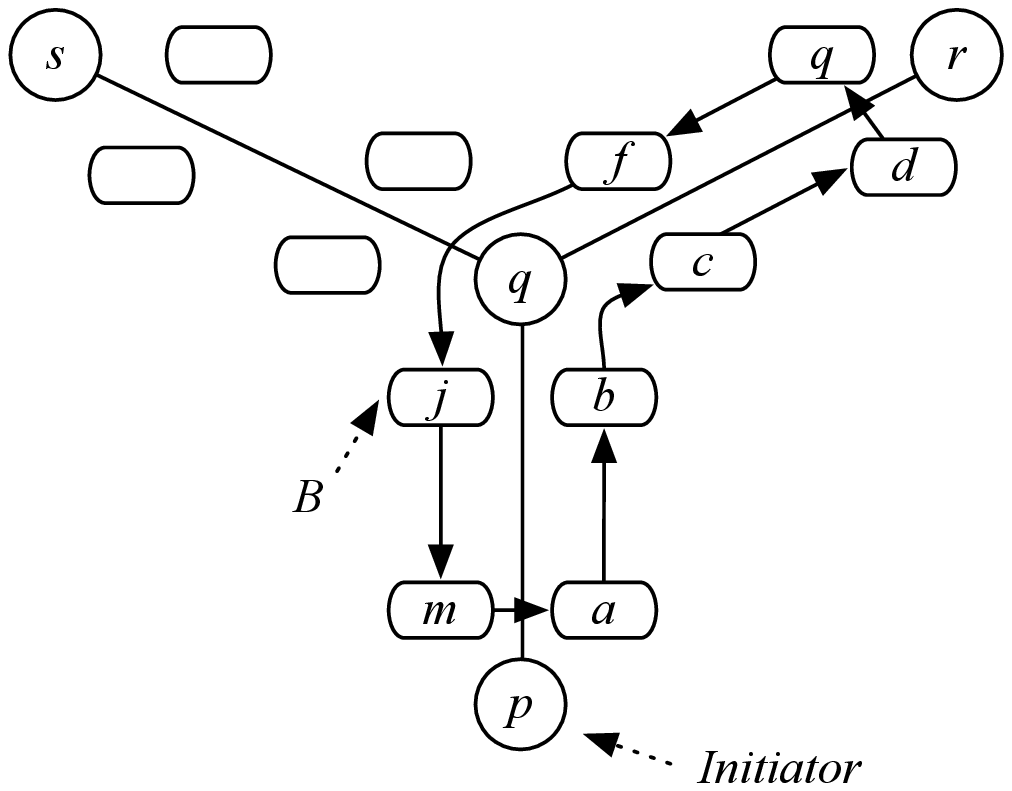,width=4cm}\\
  \textit{(a)} Instance of a Full-Cycle.
  \end{center}
 \end{minipage} \hfill
 \begin{minipage}[b]{.46\linewidth}
 \begin{center}
 \epsfig{figure=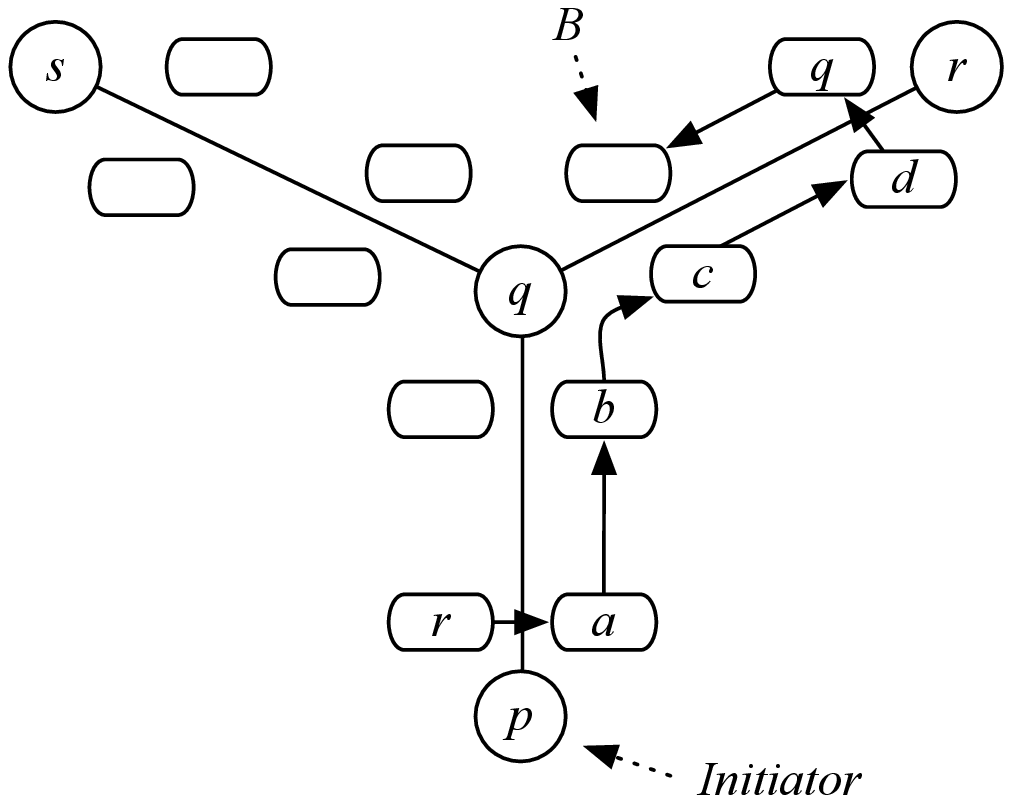,width=4cm}\\
  \textit{(b)} Free Buffer on the path
  \end{center}
 \end{minipage}\hfill
 \begin{minipage}[b]{.46\linewidth}
\begin{center}  
  \epsfig{figure=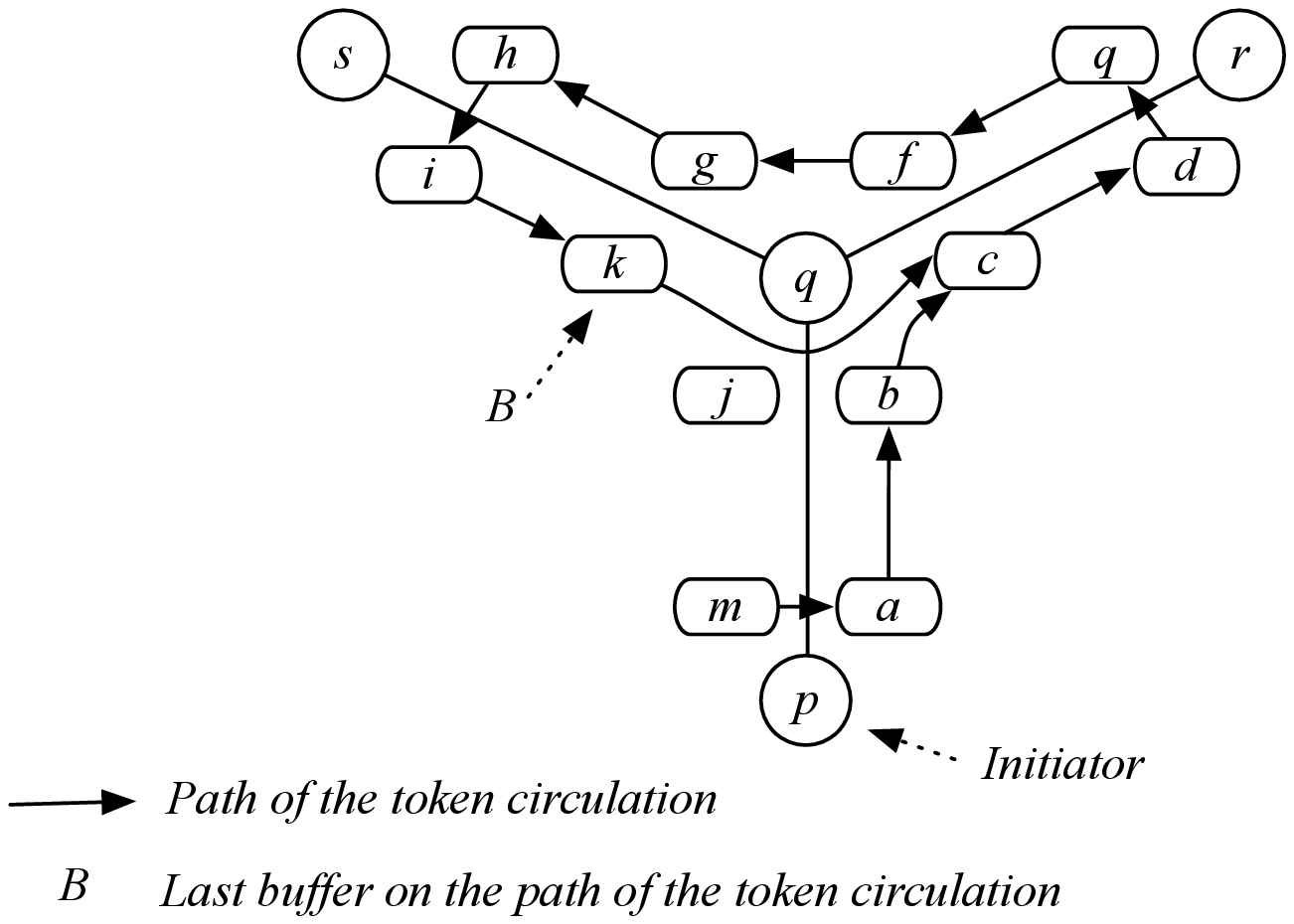,width=4.5cm}\\
  \textit{(c)} Instance of a Sub-Cycle.
  \end{center}
 \end{minipage}\hfill
  \begin{minipage}[b]{.46\linewidth}
 \begin{center}
 \epsfig{figure=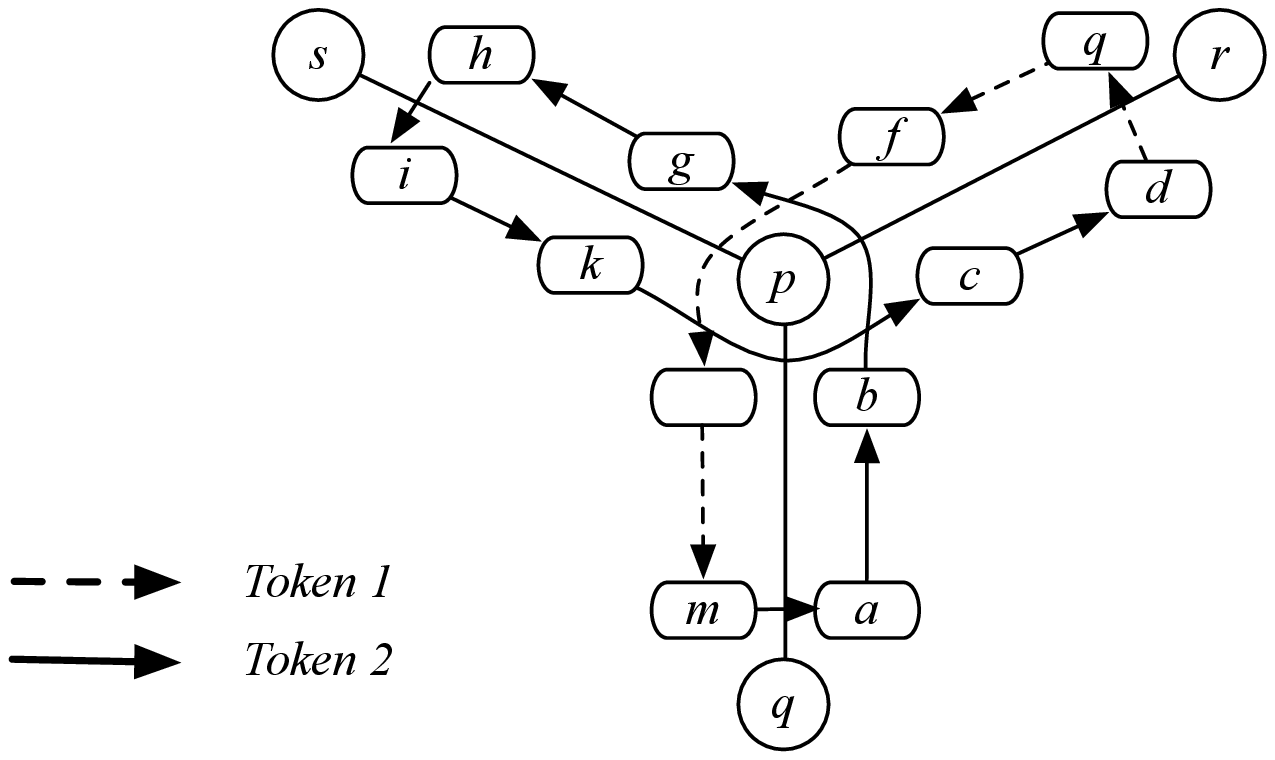,width=5cm}\\
  \textit{(d)} Token circulations deadlocked
  \end{center}
 \end{minipage}\hfill
 \caption{Instance of token circulations.\label{Tokenn}}
\end{figure}

If the token circulation has found an empty buffer (Let refer to it as $B$), the idea is to move the messages along the token circulation path to make the free slot initially on $B$ move. By doing so, we are sure that $OUT_p(q)$ becomes free. Thus $p$ can copy the message $m$ directly to $OUT_p(q)$ (Note that this action has the priority on all the other enabled actions). If the token circulation has detected a cycle then two sub-cases are possible according to the type of cycle detected: $(i)$ The case of a Full-Cycle: Note that in this case $p$ is the one that detects the cycle. The aim will be to release $OUT_p(q)$. $(ii)$ The case of a Sub-Cycle: In this case the processor containing the last buffer $B$ that is reached by the token is the one that detects the cycle (Processor $q$ in Figure \ref{Tokenn}, (c)). Note that $B$ is an input buffer. The aim in this case is to release the output buffer $B'$ by which the message $m$ in $B$ has to be forwarded to in order to meet its destination ($OUT_q(r)$ in Figure \ref{Tokenn}, (c)). Note that $B'$ is in this case part of the path of the token circulation. In both cases $(i)$ and $(ii)$, the processor that detects the cycle copies the message from the corresponding input buffer (either from $IN_p(q)$ or $B$) to its extra buffer. By doing so the processor releases its input buffer. The idea is to move messages on the token circulation's path to make the free slot that was created on the input buffer move. This ensures that the corresponding output buffer will be free in a finite time (either $OUT_p(q)$ or $B'$). Thus the message in the extra buffer can be copied in the free slot on the output buffer. Thus one cycle has been broken.


Note that many token circulations can be executed in parallel. To avoid deadlock situations between the different token circulations (refer to Figure \ref{Tokenn}, (d)), the token circulation with an identifier $id$ can use a buffer of another token circulation having the identifier $id'$ if $id<id'$. Note that by doing so, one token circulation can break the path of another one when the messages move to escort the free slot. The free slot can be then lost. For instance, in Figure \ref{Tokennn}, we can observe that the free slot that was produced by $T1$ is taking away by $T2$. By moving messages on the path of $T2$, a new cycle is created again, involving $q$ and $p$. If we suppose that the same thing happens again such as the extra buffer of $s$ becomes full and that $s$ and $p$ becomes involved again in the another cycle then the system is deadlocked and we cannot do anything to solve it since we cannot erase any valid message. Thus we have to avoid to reach such a configuration dynamically. To do so, when the token circulation finds either a free buffer or detect a cycle, it does the reverse path in order to validate its path. Thus when the path is validated no other token circulation can use a buffer that is already in the validated path. 
Note that the token is now back to the initiator. To be sure that all the path of the token circulation is a correct path (it did not merge with another token circulation that was in the initial configuration), the initiator sends back the token to confirm all the path. In another hand, since the starting configuration can be an arbitrary configuration, we may have in the system a path of a token circulation that forms a cycle. To detect and release such a situation, a value is added to the state of each buffer in the following manner: If the buffer $B_i$ has the token with the value $x$, then when the next buffer $B_{i+1}$ receive the token it will set it value at $x+1$. Thus we are sure that in the case there is a cycle there will be two consecutive buffers $B$ and $B'$ having respectively $x$ and $x'$ as a value in the path of the cycle such as $x\ne x'$. Thus this kind of situation can be detected.

\begin{figure}
 \begin{minipage}[b]{.46\linewidth}
  \begin{center}
  \epsfig{figure=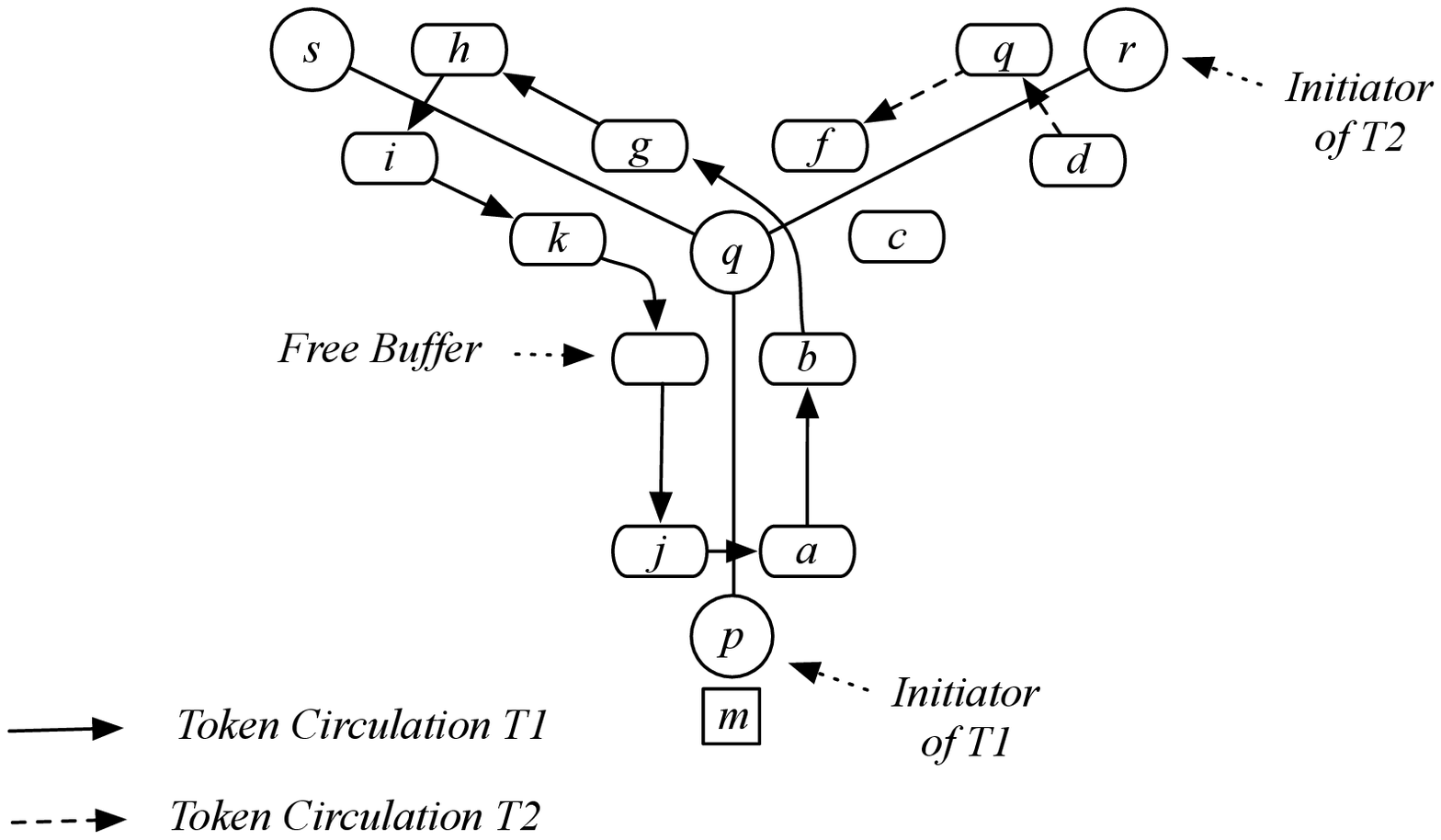,width=5.5cm}\\
  \textit{(a)} 
  \end{center}
 \end{minipage} \hfill
 \begin{minipage}[b]{.46\linewidth}
 \begin{center}
 \epsfig{figure=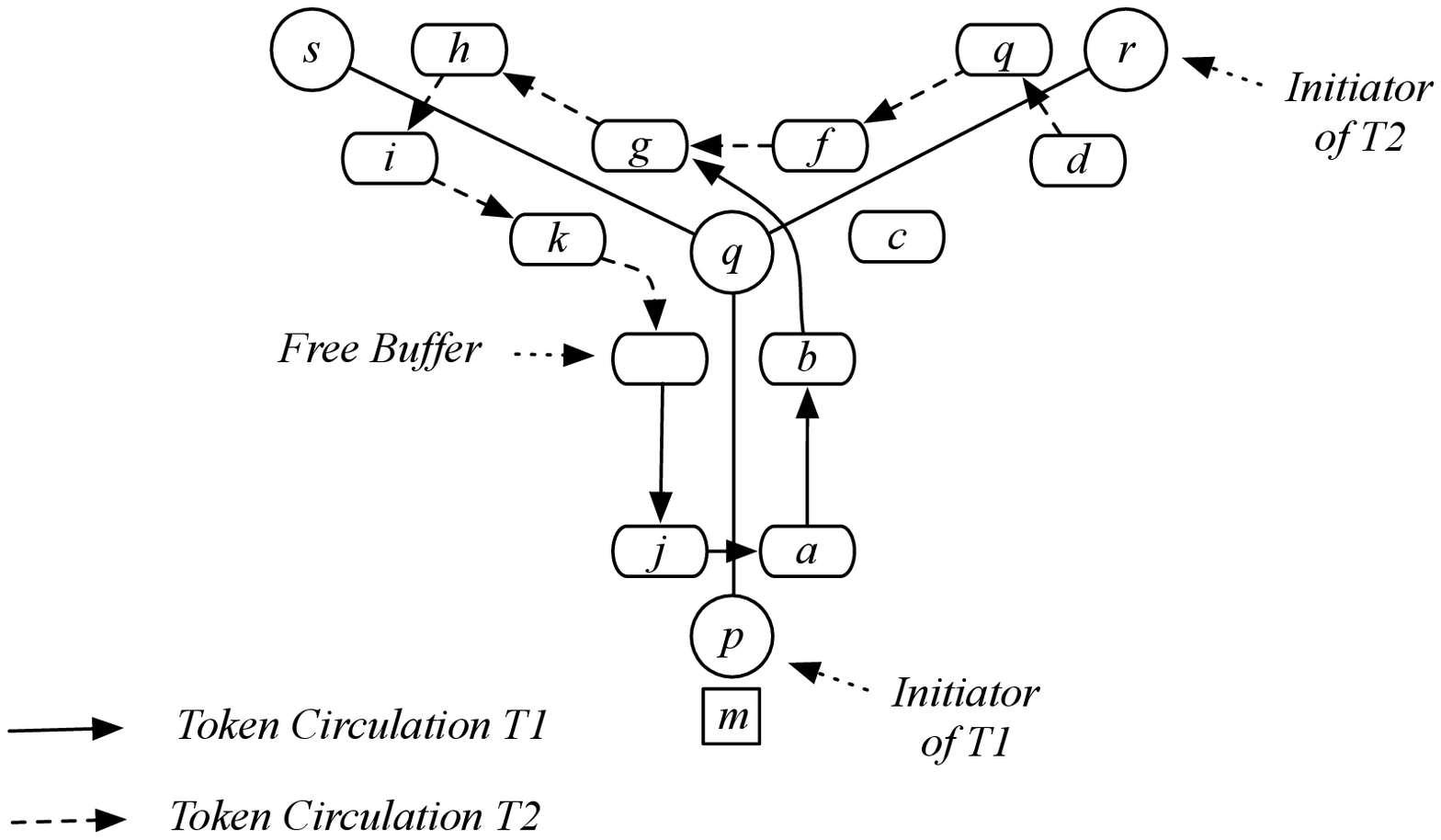,width=5.5cm}\\
  \textit{(b)}
  \end{center}
 \end{minipage}\hfill
 \begin{minipage}[b]{.46\linewidth}
\begin{center}  
  \epsfig{figure=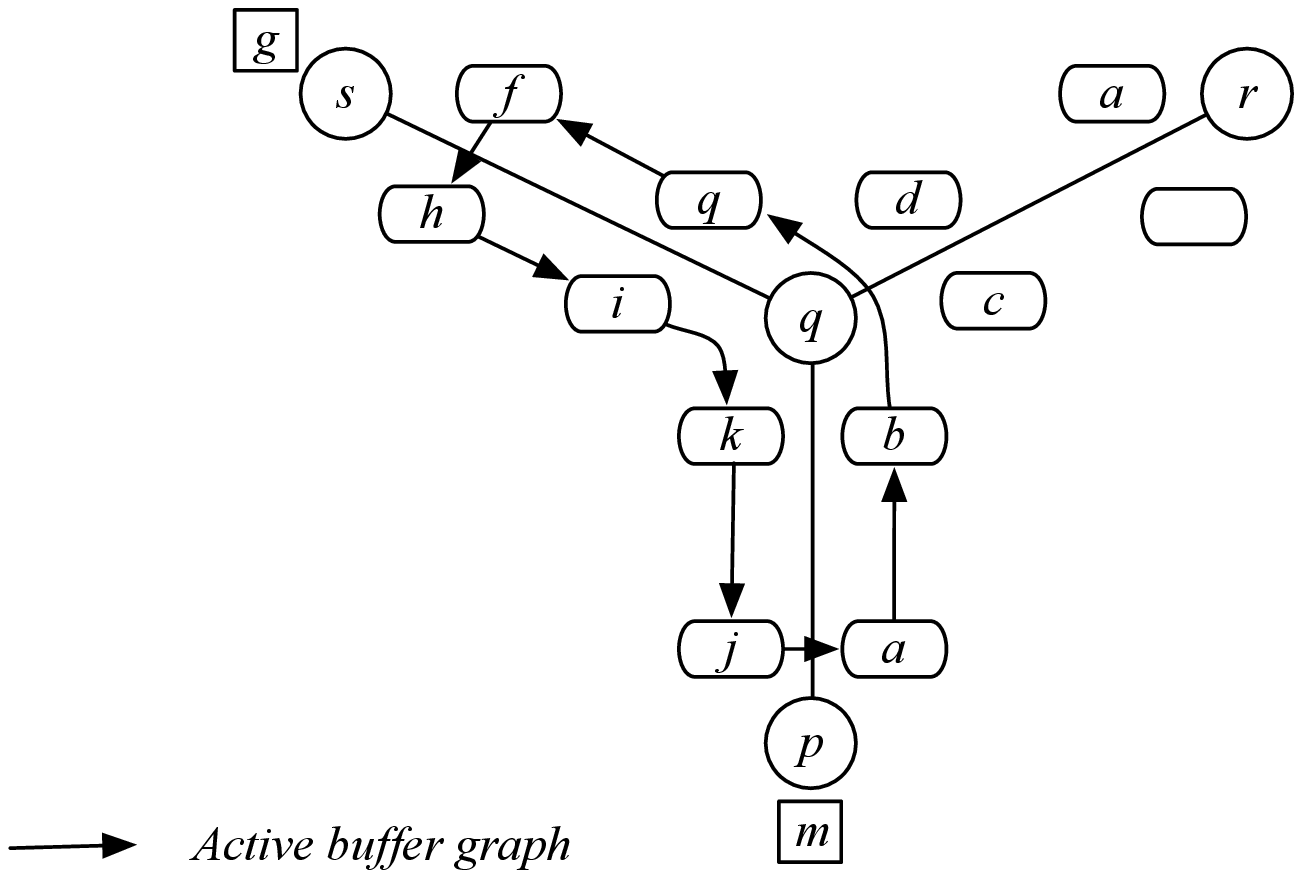,width=4.5cm}\\
  \textit{(c)} 
  \end{center}
 \end{minipage}\hfill 
 \caption{Instance of a problem.\label{Tokennn}}
\end{figure}

\subsection{Formal Description of the Solution}\label{subsec:Formal}

In this section we first define the data and variables that are used for the description of our algorithms. We then present the formal description of both the Token Circulation algorithm and the message forwarding algorithm.  

 Character {\tt '?'} in the predicates and the algorithms means {\em any value}.

\begin{itemize}
\pagebreak

\item { \textbf{Procedures}}\\
            \begin{itemize}
            \item $Next_{p}(d)$: refers to the neighbour of $p$ given by the routing tables for the destination $d$.
             \item $Deliver_{p}(m)$: delivers the message $m$ to the higher layer of $p$. \\
             \item $Choice(c)$: chooses a color for the message $m$ which is different from the color of the message that are in the buffers connected to the one that will contain $m$.\\
             \end{itemize}
 \item \textbf{Variables}
            \begin{itemize}
            \item $IN_{p}(q)$: The input buffer of $p$ associated to the link $(p,q)$.
            \item $OUT_{p}(q)$: The output buffer of $p$ associated to the link $(p,q)$.
            \item $EXT_{p}$: The Extra buffer of processor $p$.
             \item $S_{pqi}=(id,previous,next,phase,x)$: refers to the state of the input buffer of the process p on the link (p,q). id refers to the identity of the process that initiates the token circulation. previous is a pointer towards the output buffer from which the buffer $pqi$ received the token (it refers to the output buffer of $q$ on the link (q,p)). next is also a pointer that shows the next buffer that received the token from the input buffer of $p$ on the link (p,q). phase $\in \{S, V, F, C, E\}$ defines the state of the token circulation to determine which phase is executed respectively (Search, Validation, Confirm, Escort or non of these "Clean" State). $x$ is an integer which will be used in order to break incorrect cycles.
             \item $S_{pqo}=(id,previous,next,phase,x)$: As for the input buffer, $S_{pqo}=(id,previous,next,phase,x)$ refers to the state of the output buffer of the process p connected to the link (p,q). The attributes have the same meaning as previously.  
\item $prev_{pqo}$ : $q' \in N_p$ such as $S_{pq'i}=(id_{q'}, q'po, pqo, S,?)$ $\wedge$ $id_{q'}=min\{id_{q''}, q''\in N_p \wedge S_{pq''i}=(id_{q''},q'po, pqo, S,?)$\}.
              \item $Small_p$: $q \in N_p$ such as $\exists$ $q'\in N_p$, $S_{pqi}=(id_{q}, ?, pq'o, F,x)$ $\wedge$ $S_{pq'o}=(id_{q}, X, q'pi, F,z)$ $\wedge$ $X\ne pqi$ $\wedge$ $z\ne x+1$ $\wedge$ $id_{q}=min\{id_{q''}, q''\in N_p \wedge S_{pq''i}=(id_{q''},?, pro, F,x')$ $\wedge$ $S_{pro}=(id_{q''}, X', rpi, F,z')$ $\wedge$ $X'\ne pq''i$ $\wedge$ $z'\ne x'+1$.
            \end{itemize}
\item\textbf{Predicates} 
             \begin{itemize}
             \item $NO-Token_p$: $\forall$ $q \in N_p$, $S_{pqi}=(-1,NULL, NULL, C, -1)$ $\wedge$ $S_{pqo}=(-1,?,?,?)$ $\wedge$ $S_{qpo}=(-1,NULL, NULL, C, -1)$
             \end{itemize}
\item{ We define a fair pointer that chooses the actions that will be performed on the output buffer of a processor $p$. (Generation of a message or an internal transmission). }
\end{itemize}

\begin{algorithm}[H]
	\caption{Token circulation --- Initiation and Transmission \label{algo:DFTC1}}
\begin{scriptsize}

       \textbf{Token initiation}\\
$\textbf{R1}$: $Token_p(q)$ $\wedge$ $S_{pqo}=(-1,NULL,NULL,C,-1)$ $\wedge$ $S_{pqo}=S_{pqi}$ $\rightarrow$ $S_{pqi}:=(p,NULL,pqo,S,0)$, $S_{pqo}:=(p,pqi,qpi,S,1)$\\

       \textbf{Token transmission}
              \begin{itemize}
                       \item \textbf{Search phase }
                          \begin{itemize}
\item $\textbf{R2}$: $\exists$ $q,q' \in N_p$, $S_{qpo}=(id,?,pqi,S,x)$ $\wedge$ $IN_p(q)=(m,d,c)$ $\wedge$ $Next_p(d)=q'$ $\wedge$ $S_{pq'o}\ne (id,?,?,?,?)$ $\wedge$ $S_{pqi}\ne(id',?,?,V \vee F \vee E,?)$ $\wedge$ ($S_{pqi}\ne(id'',?,?,?,?)$ $\wedge$ $id''<=id$)  $\rightarrow$ $S_{pqi}:=(id,qpo,pq'o,S,x+1)$
\item $\textbf{R3}$: $\exists$ $q, q' \in N_p$, $prev_{pqo}=q'$ $\wedge$ $S_{pq'i}=(id,q'po,pqo,S,x)$ $\wedge$ ($S_{qpi}$ $\ne$ $(id,?,?,?,?)$ $\wedge$ $S_{pqo}$ $\ne$ $(id'',?,?,V \vee F \vee E,?)$ $\wedge$ $S_{pqo}$ $\ne$ $(id',?,?,?,?)$ $\wedge$ $id'<=id$ $\wedge$ $OUT_p(q)\ne \epsilon$ $\wedge$ $OUT_p(q)\ne IN_q(p)$ $\rightarrow$ $S_{pqo}:= (id,pq'i,qpi,S,x+1)$\\
                            \end{itemize}
                        \item \textbf{Validation phase}
                        \begin{itemize}
                        \item \textbf{Initiation}
                              \begin{itemize}
                                     \item  $\textbf{R4}$: $\exists q,q' \in N_p$, $prevpqo = q'$ $\wedge$ $S_{pq'i} = (id,q'po,pqo,S,x)$ $\wedge$ $S_{pqo}\ne (id'',?,?,V \vee F \vee E,?)$ $\wedge$ $S_{pqo}\ne (id',?,?,?,?)$ $\wedge$ $id' < id$ $\wedge$ $S_{qpi}=(id,X,?,S,?)$ $\wedge$ $X\ne pqo$ $\wedge$ $OUT_p(q)\ne \epsilon$ $\wedge$ $OUT_p(q)\ne IN_q(p)$ $\rightarrow$ $S_{pqo} := (id, pq'i, qpi, V, x + 1)$
                                     \item $\textbf{R5}$: $\exists$ $q, q' \in N_p$, $S_{qpo}=(id,?,pqi,S,x)$ $\wedge$ $IN_p(q)=(m,d,c)$ $\wedge$ $Next_p(d)=q'$ $\wedge$ $S_{pq'o}=(id,X,?,S,z)$ $\wedge$ $X\ne pqi$ $\wedge$ $EXT_p=\epsilon$ $\wedge$ $S_{pqi}\ne(id',?,?,V \vee F \vee E,?)$ $\wedge$ ($S_{pqi}\ne(id'',?,?,?,?)$ $\wedge$ $id''<id$) $\rightarrow$ $S_{pqi}:=(id,qpo,pq'o,V,x+1)$
                                  \item  $\textbf{R6}$: $\exists$ $q,q' \in N_p$,  $prev_{pqo}=q'$ $\wedge$ $S_{pq'}i=(id,q'po,pqo,S,x)$ $\wedge$ [($OUT_p(q)=\epsilon$ $\vee$ $OUT_p(q)=IN_q(p)$)]  $\rightarrow$ $S_{pqo}:= (id,pq'i,NULL,V,x+1)$
                                  \item $\textbf{R7}$: $\exists$ $q, q' \in N_p$, $S_{qpo}=(id,?,pqi,S,x)$ $\wedge$ $IN_p(q)=\epsilon$ $\rightarrow$ $S_{pqi}:=(id,qpo,NULL,V,x+1)$\\
                             \end{itemize}
                        \item \textbf{Transmission}
                        \begin{itemize}
                                 \item $\textbf{R8}$: $\exists$ $q, q'\in N_p$, $S_{pqo}=(id,pq'i,qpi,S,x)$ $\wedge$ $S_{qpi}=(id,pqo,?,V,x+1)$ $\wedge$ $x\ne 1$ $\wedge$ $S_{pq'i}\ne(id,?,pqo,F,x-1)$ $\rightarrow$ $S_{pqo}:=(id,?,qpi,V,x)$
                                 \item $\textbf{R9}$: $\exists$ $q,q' \in N_p$, $S_{pqi}=(id,qpo,pq'o,S,x)$ $\wedge$ $S_{pq'o}=(id,pqi,?,V,x+1)$ $\wedge$ $S_{qpo}\ne(id,?,pqi,F,x-1)$ $\rightarrow$ $S_{pqi}:=(id,qpo,pq'o,V,x)$ \\
                                 
                             \end{itemize}
                              \end{itemize}
                              \item \textbf{Confirm phase}
                              \begin{itemize}
                        \item \textbf{Initiation}
                               \begin{itemize}
                             \item $\textbf{R10}$: $\exists$ $q \in N_p$, $S_{pqo}=(p,pqi,qpi,S,1)$ $\wedge$ $S_{pqi}=(p,NULL,pqo,S,0)$ $\wedge$ $S_{qpi}=(p,pqo,?,V,2)$  $\rightarrow$ $S_{pqo}:=(p,pqi,qpi,F,1)$, $S_{pqi}:=(p,NULL,pqo,F,0)$\\
                                \end{itemize}
                        \item \textbf{Transmission}
                        \begin{itemize}
                                 \item $\textbf{R11}$: $\exists$ $q, q' \in N_p$, $S_{qpo}=(id,?,pqi,F,x)$ $\wedge$ $S_{pqi}=(id,qpo,pq'o,V,x+1)$ $\rightarrow$ $S_{pqi}:=(id,qpo,pq'o,F,x+1)$
                                  \item $\textbf{R12}$: $\exists$ $q,q' \in N_p$, $prev_{pqo}=q'$ $\wedge$ 
                                    $S_{pq'i}=(id,?,pqo,F,x)$ $\wedge$ $S_{pqo}=(id,pq'i,qpi,V,x+1)$ $\rightarrow$ $S_{pqo}:=(id,pq'i,qpi,F,x+1)$  \\	                        
                             \end{itemize}
                             \end{itemize} 
                             
             \item \textbf{Escort phase}     
                                    \begin{itemize}
                                      \item \textbf{Initiation}
                                             \begin{itemize}
                                               \item $\textbf{R13}$: $\exists$ $q\in N_p$, $S_{pqi}=(id,idle,pqo,F,0)$ $\wedge$ $S_{qpo}=(id,?,pqi,F,x)$ $\wedge$ $x\geq 3$ $\wedge$ $S_{pqo}=(id,pqi,qpi,F,1)$ $\wedge$ $EXT_p=\epsilon$ $\rightarrow$ $S_{pqi}:=(id,idle,pqo,E,0)$                                           
                                                \item $\textbf{R14}$: $Small_p=q$ $\wedge$ $\exists$ $q'\in N_p$, $S_{pqi}=(id,qpo,pq'o,F,x)$ $\wedge$ $S_{pq'o}=(id,X,q'pi,F,z)$ $\wedge$ $X\ne pqi$ $\wedge$ $z\ne x+1$ $\wedge$ $EXT_p=\epsilon$ $\wedge$ $\nexists$ $q''\in N_p$, ($S_{pq''i}=(id',NULL,Z,F,0)$ $\wedge$ $S_{Z}=(id',pq''i,?,F,1)$) $\rightarrow$ $S_{pqi}:=(id,qpo,pq'o,E,x)$ 
                                                \item $\textbf{R15}$: $\exists$ $q \in N_p$, $S_{qpo}=(id,?,pqi,F,x)$ $\wedge$ $S_{pqi}=(id,qpo,idle,V,x+1)$ $\wedge$ $IN_p(q)=\epsilon$ $\rightarrow$ $S_{pqi}:=(id,qpo,idle,E,x+1)$
                                                \item $\textbf{R16}$: $\exists$ $q,q' \in N_p$, $S_{pqi}=(id,qpo,pq'o,F,x)$ $\wedge$ $S_{pq'o}=(id,pqi,idle,V,x+1)$  $\wedge$ [$OUT_p(q)=\epsilon$ $\vee$ $OUT_p(q)=IN_q(p)$] $\rightarrow$ $S_{pq'o}:=(id,pqi,idle,E,x+1)$\\
                                                \end{itemize}     
                                                  \item \textbf{Propagation}
                                                 \begin{itemize}
                                               \item $\textbf{R17}$: $\exists q,q' \in N_p$, $S_{pqo}=(id,?,qpi,F,x)$ $\wedge$ $S_{qpi}=(id,pqo,?,E,x+1 \vee 0)$ 
                                               $\rightarrow$ $S_{pqo}:=(id,?,qpi,E,x)$
                                               \item $\textbf{R18}$: $\exists q,q' \in N_p$, $S_{pqi}=(id,qpo,pq'o,F,x)$ $\wedge$ $S_{pq'o}=(id,pqi,q'pi,E,x+1)$ 
                                               $\rightarrow$ $S_{pqi}:=(id,qpo,pq'o,E,x)$
                                               \item $\textbf{R19}$: $\exists q \in N_p$, $S_{pqo}=(id,pqi,qpi,F,1)$ $\wedge$ $S_{qpi}=(id,idle,pqo,E,0)$ $\wedge$ $S_{qpi}=(id,pqo,?,E,2)$ $\rightarrow$ $S_{pqo}:=(id,pqi,qpi,E,1)$

                                                  \end{itemize}                                                                                        
                          
                                        \end{itemize}                 
               
                  \end{itemize}

%
%
%
%
%
%
%
%
%
%
%
%
%
%
%
 
   \end{scriptsize}
\end{algorithm}

\begin{algorithm}[H]
\caption{Token Circulation --- Cleaning Phase and Correction\label{algo:DFTC2}}
\begin{scriptsize}

%
\begin{itemize}
\item \textbf{-Cleaning phase}   
                          \begin{itemize}
                               \item \textbf{Initiation}
                                             \begin{itemize}                 
                                                 \item $\textbf{R20}$: $\exists$ $q\in N_p$, $S_{pqi}=(id, NULL, pqo, E,0)$ $\wedge$ $S_{pqo}=(id, pqi, qpi, E,1)$ $\rightarrow$ $S_{pqi}:=(-1, NULL, NULL, C,-1)$ 
                                                 \item $\textbf{R21}$: $\exists$ $q,q'\in N_p$, $S_{pq'i}=(id, q'po, pqo, E,x)$ $\wedge$ $S_{pqo}=(id, X, qpi, E,z)$ $\wedge$ $X\ne pq'i$ $\rightarrow$ $S_{pq'i}:=(-1, NULL, NULL, C,-1)$ 
                                                 \item $\textbf{R22}$: $\exists$ $q\in N_p$, $S_{pqi}=(id, qpo, NULL, E,x)$ $\wedge$ $S_{qpo}=(id, ?, pqi, E,x-1)$ $\rightarrow$ $S_{pqi}:=(-1, NULL, NULL, C,-1)$ 
                                                 \item $\textbf{R23}$: $\exists$ $q,q'\in N_p$, $S_{pqo}=(id, pq'i, qpi, E,x)$ $\wedge$ $S_{pq'i}=(id, q'po, pqo, E,x-1)$ $\rightarrow$ $S_{pq'i}:=(-1, NULL, NULL, C,-1)$ 
                                                 
                                            \end{itemize}  
                              \item \textbf{Propagation}   
                                      \begin{itemize}
                                           \item $\textbf{R24}$: $\exists$ $q$ $\in N_p$, $S_{pqo}=(id,X,qpi,E,x)$ $\wedge$ $S_{X}\ne(id,?,pqo,F,x-1)$ $\wedge$ [($S_{qpi}=(id',?,?,?,?)$ $\wedge$  $id\ne id'$) $\vee$ $S_{qpi}=(-1,NULL,NULL,C,-1)$] $\rightarrow$ $S_{pqo}:=(-1, NULL, NULL, C, -1)$
                                            \item $\textbf{R25}$: $\exists$ $q$ $\in N_p$, $S_{pqi}=(id,qpo,pq'o,E,x)$ $S_{qpo}\ne(id,?,pqi,F,x-1)$ $\wedge$ [($S_{pq'o}=(id',?,?,?,?)$ $\wedge$  $id\ne id'$) $\vee$ $S_{pq'o}=(-1,NULL,NULL,C,-1)$] $\rightarrow$ $S_{pqo}:=(-1, NULL, NULL, C, -1)$
                          \end{itemize}
                             
                        \end{itemize}   

%

\item \textbf{Correction rules}
              \begin{itemize}
                           \item \textbf{Freeze Cleaning} 
                                  \begin{itemize}
                                   \item \textbf{Initiation}\\                                         
                                                 - $\textbf{R26}$: $\exists q,q' \in N_p$, $S_{pqo}=(id,pq'i,qpi,S \vee V \vee F,?)$ $\wedge$ [$S_{pq'i}=(-1,NULL,NULL,C,-1)$ $\vee$ ($S_{pq'i}=(id',?,?,?,?)$ $\wedge$ $id'\ne id$) $\vee$ ($S_{pq'i}=(id,?,M,?,?)$ $\wedge$ $M\ne pqo$)] $\rightarrow$ $S_{pqo}:=(id,pq'i,qpi,G,?)$                    
                                                                        
                                                 - $\textbf{R27}$: $\exists q,q' \in N_p$, $S_{pqi}=(id,qpo,pq'o,?,?)$ $\wedge$ ($S_{qpo}=(-1,NULL,NULL,C,-1)$ $\vee$ ($S_{qpo}=(id',?,?,?,?)$ $\wedge$ $id'\ne id$)) $\rightarrow$ $S_{pqi}:=(id,qpo,pq'o,G,?)$                            
                                                                   
                                                 - $\textbf{R28}$: $S_{pqi}=(p,NULL,pqo,?,x)$ $\wedge$ $x>0$ $\rightarrow$ $S_{pqi}:=(p,NULL,pqo,G,x)$
                                                                      
                                                   - $\textbf{R29}$: $\exists$ $q,q' \in N_p$, $S_{pqi}=(id,?,pq'o,?,x)$ $\wedge$ $S_{pqo}=(id,pq'i,qpi,?,z)$ $\wedge$ $z\ne x+1$ $\rightarrow$ $S_{pqi}:=(id,?,pq'o,G,x)$ 
                                                                                                   
                                                    - $\textbf{R30}$: $\exists$ $q \in N_p$, $S_{pqi}=(id,qpo,pqo,?,x)$ $\wedge$ $S_{qpo}=(id,?,pqi,?,z)$ $\wedge$ $z\ne x+1$ $\rightarrow$ $S_{pqi}:=(id,qpo,pqo,G,x)$    
                                                    
                                                     - $\textbf{R31}$: $\exists$ $q \in N_p$, [($S_{pqo}=(id,?,qpi,S,x)$ $\wedge$ $S_{qpi}=(id,pqo,?,F \vee E,x+1)$) $\vee$ ($S_{pqo}=(id,?,qpi,F,x)$ $\wedge$ $S_{qpi}=(id,pqo,?,S,x+1)$) $\vee$ ($S_{pqo}=(id,?,qpi,V,x)$ $\wedge$ $S_{qpi}=(id,pqo,?,E \vee S,x+1)$)] $\rightarrow$ $S_{pqo}:=(id,?,qpi,G,x)$
                                                     
                                                      - $\textbf{R32}$: $\exists$ $q,q' \in N_p$, [($S_{pqi}=(id,?,pq'o,S,x)$ $\wedge$ $S_{pq'o}=(id,pqi,?,F \vee E,x+1)$) $\vee$ ($S_{pqi}=(id,?,pq'o,F,x)$ $\wedge$ $S_{pq'o}=(id,pqi,?,S,x+1)$) $\vee$ ($S_{pqi}=(id,?,pq'o,V,x)$ $\wedge$ $S_{pq'o}=(id,pqo,?,E \vee S,x+1)$)] $\rightarrow$ $S_{pqi}:=(id,?,pq'o,G,x)$\\     
                                                                                                   
                                   \item \textbf{Propagation}  \\                                   
                                                  - $\textbf{R33}$:$\exists$ $q,q' \in N_p$, $S_{qpo}=(id,?,pqi,G,?)$ $\wedge$ $S_{pqi}=(id,qpo,pq'o,S \vee V \vee F \vee E,?)$ $\rightarrow$ $S_{pqi}:=(id,qpo,pq'o,G,?)$  
                                                                                              
                                                  - $\textbf{R34}$: $\exists$ $q,q' \in N_p$, $prev_{pqo}=q$ $\wedge$ $S_{pqi}=(id,?,pq'i,G,?)$ $\wedge$ $S_{pq'o}=(id,qpo,q'pi,S \vee V \vee F \vee E,?)$ $\rightarrow$ $S_{pq'o}:=(id,pqi,q'pi,G,?)$\\
                                                 
                                   \item \textbf{Cleaning}\\                                   
                                                   - $\textbf{R35}$: $\exists q,q' \in N_p$, $S_{pqi}=(id,qpo,pq'o,G,x)$ $\wedge$ [$S_{pq'o}=(-1,NULL,NULL,C,-1)$ $\vee$ ($S_{pq'o}=(id',?,?,?,?)$ $\wedge$ $id'\ne id$) $\vee$ ($S_{pq'o}=(id,?,qpi,G,z)$ 
                                                   $\rightarrow$ $S_{pqi}:=(-1,NULL,NULL,C,-1)$ 
                                                                                                     
                                                   - $\textbf{R36}$: $\exists q,q' \in N_p$, $S_{pqo}=(id,pq'i,qpi,G,x)$ $\wedge$ [$S_{qpi}=(-1,NULL,NULL,C,-1)$ $\vee$ ($S_{qpi}=(id',?,?,?,?)$ $\wedge$ $id'\ne id$) $\vee$ ($S_{qpi}=(id,pqo,?,G,z)$ 
                                                   $\rightarrow$ $S_{pqo}:=(-1,NULL,NULL,C,-1)$\\

                                  \end{itemize}
                             
                      \item $\textbf{R37}$: $\exists$ $q,q' \in N_p$, $S_{pqi}=(id,?,pq'o,G,x)$ $\wedge$ $S_{pq'o}=(id,pqi,q'pi,G,z)$ $\wedge$ $z\ne x+1$ $\rightarrow$ $S_{pqi}:=(-1,NULL,NULL,C,-1)$
                      \item $\textbf{R38}$: $\exists$ $q \in N_p$, $S_{pqi}=(id,qpo,pqo,G,x)$ $\wedge$ $S_{qpo}=(id,?,pqi,G,z)$ $\wedge$ $z\ne x+1$ $\rightarrow$ $S_{pqi}:=(-1,NULL,NULL,C,-1)$                                          
                       \item $\textbf{R39}$: $Token_p(q)$ $\wedge$ $S_{pqi}=(p,?,?,?,?)$ $\rightarrow$ $Token_p(q):=false$      
                       \item $\textbf{R40}$: $\exists$ $q,q' \in N_p$, $S_{pqi}=(id,qpo,pq'o,F,x)$ $\wedge$ $S_{pq'o}=(id,X,q'pi,S \vee V,z)$ $\wedge$ $z\ne x+1$ $\rightarrow$ $S_{pqi}:=(-1,NULL,NULL,C,-1)$   
                       \item $\textbf{R41}$: $\exists$ $q \in N_p$, $S_{pqi}=(id,qpo,NULL, S \vee V \vee F ,x)$ $\wedge$ $IN_p(q)\ne \epsilon$ $\rightarrow$ $S_{pqi}:=(-1,NULL,NULL,C,-1)$ 
                       \item$\textbf{R42}$: $\exists$ $q,q' \in N_p$, $S_{pqo}=(id,pq'i,NULL,S \vee V \vee F,x)$ $\wedge$ $OUT_p(q)\ne \epsilon$ $\rightarrow$ $S_{pqo}:=(-1,NULL,NULL,C,-1)$   
                        \item $\textbf{R43}$ $\exists$ $q$ $\in N_p$, $S_{pqo}=(id,?,qpi,V \vee F,x)$ $\wedge$ [($S_{qpi}=(id',?,?,?,?)$ $\wedge$  $id\ne id'$) $\vee$ $S_{qpi}=(-1,NULL,NULL,C,-1)$] $\rightarrow$ $S_{pqo}:=(-1, NULL, NULL, C, -1)$
                        \item  $\textbf{R44}$   $\exists$ $q$ $\in N_p$, $S_{pqi}=(id,qpo,pq'o,V \vee F \vee E,x)$ $\wedge$ [($S_{pq'o}=(id',?,?,?,?)$ $\wedge$  $id\ne id'$) $\vee$ $S_{pq'o}=(-1,NULL,NULL,C,-1)$] $\rightarrow$ $S_{pqo}:=(-1, NULL, NULL, C, -1)$                    
                       
 \end{itemize}
      \end{itemize}
\end{scriptsize}
\end{algorithm}

\begin{algorithm}[H]
\caption{Message Forwarding \label{algo:MF}}
\begin{scriptsize}

   \begin{itemize}
       \item{\textbf{Message generation (For every processor) }}
       
$\textbf{R'1}$: $Request_{p}$ $\wedge$ $ Next_{p}(d)=q$ $\wedge$ [$OUT_{p}(q)=\epsilon$ $\vee$ $OUT_{p}(q)=IN_{q}(p)$] $\wedge$ $NO-Token$ $\rightarrow$ $OUT_{p}(q):=(m,d,choice(c))$, $Request_{p}:=false$.\\

       \item{\textbf{Message consumption (For every processor) }}
       
$\textbf{R'2}$: $\exists q \in N_{p}$, $IN_p(q)=(m,d,c)$ $\wedge$ $d=p$ $\wedge$ $OUT_{q}(p) \ne IN_{p}(q)$  $\rightarrow$ $deliver_{p}(m)$, $IN_{p}(q):=OUT_{q}(p)$.\\
     

        \item{\textbf{Internal transmission}}
 
         $\textbf{R'3}$: $\exists q,q' \in N_{p}$, $IN_p(q)=(m,d,c)$ $\wedge$ $d\ne p$ $\wedge$ $Next_p(d)=q'$  $\wedge$ $q'\ne q$  $\wedge$ [$OUT_{p}(q')=\epsilon$ $\vee$ $OUT_{p}(q')=IN_{q'}(p)$] $\wedge$ $OUT_q(p)\ne IN_p(q)$ $\wedge$ $NO-Token$ $\rightarrow$ $OUT_{p}(q'):=(m,d,choice(c))$, $IN_{p}(q):=OUT_{q}(p)$.\\
         
          $\textbf{R'4}$: $\exists$ $q,q' \in N_p$, $IN_p(q')=(m,d,c)$ $\wedge$ $OUT_{q'}(p)\ne IN_p(q')$ $\wedge$ [$OUT_p(q)=\epsilon$ $\vee$ $OUT_p(q)=IN_q(p)$] $\wedge$ $S_{pqo}=(id,pq'i,qpi,E,x+1)$ $\wedge$ $S_{pq'i}=(id,q'po,pqo,F,x)$ 
          $\rightarrow$ $OUT_p(q):=(m,d,choice(c))$, $IN_p(q'):=OUT_{q'}(p)$         \\
         
         

         \item{\textbf{Message transmission from $q$ to $p$ }} 

  $\textbf{R'5}$: $IN_{p}(q)=\epsilon$ $\wedge$ $OUT_{q}(p)=(m,d,c)$ $\wedge$ $q\ne d$ $\wedge$ $NO-Token$ $\rightarrow$ $IN_{p}(q):=OUT_{q}(p)$.\\ 
  
  $\textbf{R'6}$: $ \exists q\in N_p$, $IN_p(q)=\epsilon$ $\wedge$ $OUT_{q}(p)=(m,d,c)$ $\wedge$ $q\ne d$ $\wedge$ $S_{pqi}=(id,qpo,?,E,x+1)$ $\wedge$ $S_{qpo}=(id,?,pqi,E,x)$ $\rightarrow$ $IN_p(q):=OUTq(p)$\\
 
         
             
         \item{\textbf{Erasing a message after its transmission }}
         
$\textbf{R'7}$: $\exists q \in N_{p}$, $OUT_{p}(q)=IN_{q}(p)$ $\wedge$ ($\forall q' \in N_{p}\setminus \{q\}$,
$IN_{p}(q')=\epsilon$ $\vee$ $(IN_{p}(q')=(m,d,c) \wedge Next_p(d)\ne q$)) $\wedge$ $NO-Token$ $\rightarrow$ $OUT_{p}(q):=\epsilon$ \\

            
         \item{\textbf{Erasing a message after its transmission (For the leaf processors) }}
        
 $\textbf{R'8}$: $N_{p}=\{q\}$ $\wedge$ $OUT_{p}(q)=IN_{q}(p)$ $\wedge$ ($IN_{p}(q)=\epsilon$ $\vee$ $(IN_{p}(q)=(m,d,c) \wedge Next_p(d)\ne q$)) $\wedge$ $NO-Token$  $\rightarrow$ $OUT_{p}(q):=\epsilon$\\        
        
             
        \item{\textbf{Road change}}
       
$\textbf{R'9}$: $\exists$ $q \in N_p$, $IN_p(q)=(m,d,c)$ $\wedge$ $Next_p(d)=q$ $\wedge$ $OUT_p(q)=\epsilon$ $\vee$ $OUT_p(q)=IN_q(p)$ $\rightarrow$ $OUT_p(q):=IN_p(q)$, $IN_p(q):=OUT_q(p)$\\

$\textbf{R'10}$: $\exists$ $q \in N_p$, $IN_p(q)=(m,d,c)$ $\wedge$ $Next_p(d)=q$ $\wedge$ $OUT_p(q) \ne \epsilon$ $\wedge$ $OUT_p(q)\ne IN_q(p)$ $\wedge$ $EXT_p=\epsilon$ $\wedge$ $\nexists$ $q'\in N_p$, $S_{pq'i}=(id,?,pq'o,?,0)$ $\rightarrow$ $Token_p(q):=true$\\



$\textbf{R'11}$: $\exists$ $q\in N_p$, $S_{pqi}=(id,idle,pqo,F,0)$ $\wedge$ $S_{qpo}=(id,?,pqi,F,x)$ $\wedge$ $x\geq 3$ $\wedge$ $S_{pqo}=(id,pqi,qpi,F,1)$ $\wedge$ $EXT_p= \epsilon$ $\rightarrow$ $EXT_p:=IN_p(q)$, $IN_p(q):=OUTq(p)$\\

$\textbf{R'12}$: $Small_p=q$ $\wedge$ $\exists$ $q'\in N_p$, $S_{pqi}=(id,qpo,pq'o,F,x)$ $\wedge$ $S_{pq'o}=(id,X,q'pi,F,z)$ $\wedge$ $x\ne pqi$ $\wedge$ $z\ne x+1$ $\wedge$ $S_{qpo}=(id,?,pqi,F,x-1)$ $\wedge$ $\nexists$ $q''\in N_p$, ($S_{pq''i}=(id',NULL,Z,F,0)$ $\wedge$ $S_{Z}=(id',pq''i,?,F,1)$) $\rightarrow$ $EXT_p:=IN_p(q)$, $IN_p(q):=OUTq(p)$\\


$\textbf{R'13}$: $\exists$ $q \in N_p$ $S_{pqi}=(id,NULL,pqo,E,0)$ $\wedge$ $S_{pqo}=(id,pqi,qpi,F,1)$ $\wedge$ $S_{qpi}=(id,pqo,?,E,2)$ $\wedge$ $EXT_p \ne \epsilon$ $\wedge$ ($OUT_p(q)=\epsilon$ $\vee$ $OUT_p(q)=IN_q(p)$) $\rightarrow$ $OUT_p(q):=EXT_p$, $EXT_p:=\epsilon$ \\

$\textbf{R'14}$: $\exists$ $q,q' \in N_p$ $S_{pq'i}=(id,q'po,pqo,E,x)$ $\wedge$ $S_{pqo}=(id,X,qpi,F,z)$ $\wedge$ $X \ne pq'i$ $\wedge$ $z\ne x+1$ $\wedge$ $S_{qpi}=(id,pqo,?,E,z+1)$ $\wedge$ $EXT_p \ne \epsilon$ $\wedge$ ($OUT_p(q)=\epsilon$ $\vee$ $OUT_p(q)=IN_q(p)$) $\rightarrow$ $OUT_p(q):=EXT_p$, $EXT_p:=\epsilon$ \\

 \item{\textbf{Correction Rules}}
 
 $\textbf{R'15}$: $EXT_p \ne \epsilon$ $\wedge$ ($NO-Token$ $\wedge$ ($\forall$ $q \in N_p$, $S_{pqi}\ne(id,qpo, ?, E)$) $\wedge$ ($\exists$ $q$ $\in$ $N_p$, $S_{pqi}=(id,NULL,pqo, E, 0)$ $\wedge$ $S_{pqo}=(id,pqi,qpi, E, 1)$ $\wedge$ $OUT_p(q)\ne \epsilon$ $\wedge$ $OUT_p(q)\ne IN_q(p)$) $\rightarrow$ $EXT_p:=\epsilon$\\
 
 $\textbf{R'16}$: $EXT_p \ne \epsilon$ $\wedge$ ($NO-Token$ $\wedge$ ($\forall$ $q \in N_p$, $S_{pqi}\ne(id,qpo, ?, E)$) $\wedge$ ($\exists$ $q, q'$ $\in$ $N_p$, $S_{pq'i}=(id,?,pqo, E,x)$ $\wedge$ $S_{pqo}=(id,X,qpi, E, z)$ $\wedge$ $X \ne pq'i$ $\wedge$ $z\ne x+1$ $\wedge$  $OUT_p(q)\ne \epsilon$ $\wedge$ $OUT_p(q)\ne IN_q(p)$) $\rightarrow$ $EXT_p:=\epsilon$\\

$\textbf{R'17}$: $Token_p(q)=true$ $\wedge$ $IN_p(q)=\epsilon$ $\vee$ $IN_p(q)=(m,d,c)$ $\wedge$ $Next_(d)\ne q$  $\rightarrow$ $Token_p(q)=false$\\

         
    

  \end{itemize}
  \end{scriptsize}
  \end{algorithm}

\subsection{Proof of correctness}\label{subsec:Proofs}
We prove in this section the correctness of our algorithm. The idea of the proofs is the following: we first show that no valid message is deleted from the system unless it is delivered to its destination. We then show that each buffer is infinitely often free, thus neither deadlocks nor starvation appear in the system. We finally show that every valid message is delivered to its destination once and only once in a finite time. Before detailing the proofs, let define some notions that will be used later.

\begin{definition}
Let $B_1$ and $B_2$ be two buffers and $p$, $q$ and $q'$ be processors in the network such that one of those properties holds:
\begin{itemize}
\item $\exists$ $p, q, q'$ such as $B_1=IN_p(q)$ $\wedge$ $B_2=OUT_p(q')$
\item $\exists$ $p, q$ such as $B_1=OUT_p(q)$ $\wedge$ $B_2=IN_q(p)$
\end{itemize}

$B2$ is called the successor of $B1$ denoted by $B_1\mapsto B_2$ if and only if $S_{B_1}=(id,?,B_2,?,x)$ $\wedge$ $S_{B_2}=(id,B_1,?,?,x+1)$
\end{definition}

%

\begin{definition}
A sequence of $k$ buffers $B_1\mapsto B_2 \mapsto ... \mapsto B_k$ starting from $B_1$ is called an abnormal sequence if the following property holds:

$S_{B_1}=(id,?,?,?,?)$ $\wedge$ ($B_1=IN_p(q)$ $\vee$ $B_1=OUT_p(q)$) $\wedge$ $id\ne p$

\end{definition}



A buffer $B$ is said to be cleared if $S_B=(-1,NULL,NULL,C,-1)$. In the same manner, a sequence is
said to be cleared, if all the buffers part of it becomes cleared in a finite time. 

Let us state the following lemma:

\begin{lem}\label{AbnSequence}
If the configuration contains an abnormal sequence $S_1$ of buffers $B_1 \mapsto B_2 \mapsto ...\mapsto B_k$, then $S_1$ will be cleared in a finite time.
\end{lem}

\begin{proof}
Since $S_1$ is an abnormal sequence. There is one processor $p$ that has sent a token which it did not receive from any other processor ($p$ is not the initiator). This processor is the one with the buffer $B_1$. Note that $p$ will be able to detect such a situation and either $R26$ or $R27$ will be enabled on $p$. When $p$ executes either $R26$ or $R27$ the Freeze cleaning phase is initiated, thus, $S_{B_1}=(id, NULL, B_2, G,x)$. Either $R33$ or $R34$ becomes enabled on the process that has $B_2$ as a buffer. Once one of these two rules is executed $S_{B_2}=(id, B_1, B_3, G,?)$, $B_3$ will set also its state to the Freeze cleaning phase, and so on. Thus all the buffers that are in the sequence $B_1 \mapsto B_2 \mapsto ...\mapsto B_k$ will be in the freeze phase $G$. Note that on the process $p'$ that has $B_k$ as a buffer (note that $B_k$ is the last buffer of the sequence), either $R35$ or $R36$ is enabled on $p'$. Once $p'$ executes one of these rules $B_k$ is cleared. $R35$ or $R36$ becomes then enabled on the process that has $B_{k-1}$ as a buffer. Thus when one of these rules is executed $B_{k-1}$ is cleared as well and so on. Thus we are sure that after a finite time each buffer that is in $S_1$ will clear its state and the Lemma holds. Note that the sequence $S_1$ can be broken (another token circulation with a smallest identifier can use one of the buffer of $S_1$). Note that in this case $S_1$ is divided in two sub abnormal sequence. Each sub abnormal sequence will behave on its own. Thus the buffers in each sub abnormal sequence will be cleared in a finite time. Observe that if in the sequence $S_{B_k}=(id, B_{k-1},B_1,?,z)$ and $S_{B_1}=(id, B_k,B_2,?,x)$ then we are sure that $z\ne x+1$. In this case too, the processor having $B_1$ as a buffer will be able to detect such a situation and initiates the freeze cleaning phase as previously. Thus the lemma holds.
\end{proof}

Let $p$, $q$ and $q'$ be processors such as $q$, $q'$ $\in N_p$, we state the following Lemma:

\begin{lem}\label{EndFreeBuffer}
If a valid message $m$ is copied in $EXT_p$ from $IN_p(q)$ in order to be copied later in $OUT_p(q')$, then when $S_{pq'o}=(id,?,?,E,?)$, $EXT_p$ is free.
\end{lem}

\begin{proof}
Since the message $m$ is copied in $EXT_p$, $m$ is in the wrong direction. $IN_p(q)$ containing $m$ is part of a complete token circulation $T$ \ie a token circulation that validated and confirmed all its path (Recall that no message can be generated in the presence of a token circulation (see Rule $R'1$) and, if an abnormal token circulation reaches $IN_p(q)$ after the generation of the message $m$, we are sure that the path of such a token will never be confirmed moreover all the buffers part of it will clear their state in a finite time (refer to Lemma \ref{AbnSequence})). To simplify the explanation let us define $T$ as follow: $T= B_1\mapsto B_2 \mapsto ... \mapsto B_k$. Note that $IN_p(q)$ (mentioned in the lemma) can be either $B_1$ (in the case of a full-cycle) or $B_k$ (in the case of a sub-cycle). In the following we will consider only the case of a full-cycle (the same reasoning holds for the sub-cycle case). We show that there is a synchrony between the forwarding and the token circulation algorithms. When the token circulation confirmed  all its path (all the buffer part of $T$ have their State attribute set at $E$), $R'11$ and $R13$ becomes enabled on $p$. Recall that in this case $p$ executes both of them, thus $m$ is copied in $EXT_p$, $B_1=B_K$ and $S_{B_1}=(id, NULL,B_2,E,0)$ ($B_K$ becomes a free buffer). $R17$ becomes enabled on the processor with the buffer $B_{k}$. When the rule is executed $S_{B_{k}}=(id, B_{k-1},B_1,E,x)$. Observe that $B_K$ is an output buffer whereas $B_{k-1}$ which is an input buffer. Both $R'4$ and $R18$ become enabled on the processor with the two buffers $B_k$ and $B_{k-1}$. When both rules are executed $S_{B_{k-1}}=(id, B_{k-2},B_k,E,x)$ and $B_{k-2}=B_{k-1}$. Note that the same situation as the first one appear. We can observe that when an output buffer $B$ part of $T$ is free with the state $S_B=(id,?,?,F,z)$, $R17$ is enabled on the processor $p'$ with the buffer $B$, thus the state of $B$ will be set to $S_B=(id,?,?,E,z)$ and notice that $B$ remains free. Thus on $p'$ two rules will be enabled (the internal transmission ($R'4$) and the propagation of the escort phase ($R18$)), when both are executed we retrieve the same situation with another empty output buffer, and so on. Hence we are sure that on the processor $p$, $R'13$ and $R19$ will be enabled at the same time. When both rules are executed, $EXT_p$ is free and $S_{B_2}=(id,B_1,B_3,E,1)$ where $B_2$ refers to $OUT_p(q)$ and the lemma holds.

\end{proof}

	
We can now detect in some cases if the message in the extra buffer is invalid (it was in the initial configuration). Note that the algorithm deletes a message only in such cases (when we are sure that the message in the extra buffer is invalid), refer to Rules $R'15$ and $R'16$. Thus we have the following Theorem:

\begin{theorem}\label{NoMsgDeleted}
No valid message is deleted from the system unless it is delivered to its destination.
\end{theorem}

\begin{proof}
The proof is by contradiction: we first suppose that there is a message $m$ that is deleted without being delivered to its destination.\\
\begin{itemize}
\item By construction of $R'3$ and $R'4$, this cannot be a result of an internal forwarding since the message $m$ is first of copied in $OUT_p(q)$ before being erased from $IN_p(q′)$. Note that these two rules are enabled only if $OUTp_(q)=IN_q(p)$ or $OUT_p(q) =\epsilon$. Hence when the message $m$ is copied in the $OUT_p(q)$ no message is deleted.

\item By the construction of Rule $R'5$ and $R'6$, the message is only copied in $IN_p(q)$ and not deleted from the $OUT_q(p)$. Note that $IN_p(q)$ is empty. Thus no message is erased in this case. 

\item By the construction of rule $R'7$ and $R'8$, the message in $OUT_p(q)$ is deleted. However note that in this case $OUT_p(q) = IN_q(p)$. Thus there is still a copy in the system of the message erased.

\item By the construction of rule $R'9$, $R'13$ and $R'14$. The message is first copied in $OUT_p(q)$ (note that $OUT_p(q)$ is in this case empty) before being erased. 

\item The same holds for $R'11$ and $R'12$, the message in $IN_p(q)$ is first copied in $EXT_p$ (note that $EXT_p$ is in this case empty), before being erased. Thus there is still a copy in the system of such a message.  

\item Concerning $R'15$ and $R'16$, according to Lemma \ref{EndFreeBuffer}. If one of these rules ($R'15$ or $R'16$) is enabled in $p$ then we are sure that the message in $EXT_p$ is an invalid message. Thus when the processor $p$ executes on of them, no valid message is deleted.  \\

\end{itemize}
We can deduce from all the cases above that no valid message is deleted unless it is delivered to its destination, hence the lemma holds.

\end{proof}


We now show in Lemma \ref{EXT} that the extra buffer of any processor $p$ cannot be
infinitely continuously busy (Recall that the extra buffer is used to solve the problem of
deadlocks). 

\begin{lem}\label{EXT}
If the extra buffer of the processor $p$  ($EXT_p$) contains a message, then this buffer becomes free after a finite time.
\end{lem}  

\begin{proof}
Suppose that the extra buffer contains a message. The cases below are possible:
\begin{enumerate}
\item \label{NOE} There is no token circulation including an input buffer of $p$. In this case the message that is in the extra buffer $EXT_p$ is deleted by the processor $p$ by executing either $R'15$ or $R'16$ and the lemma holds.
\item \label{NOC}There is no $q$, $q'$ $\in N_p$ such that either $(i)$ $S_{pqi}=(id,NULL,pqo,State,0)$ and $S_{pqo}=(id,pqi,qpi,State',1)$ or, $(ii)$ $S_{pqi}=(id,qpo,pq'o,State,x)$ and $S_{pq'o}=(id,X,q'pi,State',z)$ and $z\ne x+1$ hold. In this case too the message that is in the extra buffer $EXT_p$ is deleted by the processor $p$ by executing either $R'15$ or $R'16$ and the lemma holds. Observe that if either $(i)$ or $(ii)$ holds for $q$, $q'$ $\in N_p$ such that $State\ne E$ and $State'\ne F$ then in this case too the message in $EXT_p$ is deleted.

\item if there exists $q$, $q'$ $\in N_p$ such that either $(i)$ $S_{pqi}=(id,NULL,pqo,E,0)$ and $S_{pqo}=(id,pqi,qpi,F,1)$ or, $(ii)$ $S_{pqi}=(id,qpo,pq'o,E,x) and S_{pq'o}=(id,X,q'pi,F,z)$ and $z\ne x+1$ holds then the following two sub-cases are possible:
        
        \begin{itemize}
        \item The token circulation is an abnormal sequence. In this case, we are sure that all the buffers part of it will clear their state (refer to Lemma \ref{AbnSequence}). Thus we retrieve case \ref{NOE}. 
         \item The token circulation is a valid token circulation. In this case we are sure that the state of $OUT_p(q)$ (resp, $OUT_p(q')$) will be set at $(id,?,?,E,?)$. Thus if $OUT_p(q)$ (resp $OUT_p(q')$) is free then the message in $EXT_p$ is copied in it (refer to Rules $R'13$ and $R'14$). If it is not free the message in $EXT_p$ is deleted (refer to $R'15$ and $R'16$).
        \end{itemize}

\end{enumerate}
From the cases above, we can deduce that if $EXT_p$ is occupied then it will be cleared in a finite time.
\end{proof}


\begin{lem}\label{TokenClear}
If there is a Token Circulation that validates all its sequence, then all the buffers part of it will clear their state in a finite time.
\end{lem}

\begin{proof}
Let refer to the Token Circulation as $T=B_1\mapsto B_2 \mapsto ... \mapsto B_k$. Since the token circulation validates all its path then either it found a free buffer or detects a cycle. Note that ($i$) if its' non of these cases hold then the last buffer of the sequence $B_k$ will clear its state ($R41$ or $R42$ is executed). $B_{k-1}$ then does the same and so on. Otherwise, ($ii$) the Confirm phase is initiated by the initiator of $T$ and we can easily show that all the buffers part of $T$ will update their state to the confirm phase in a finite time. The escort phase is then initiated by either the initiator of $T$ (in the case of a full-cycle) or by the processor that has $B_k$ as a buffer. Observe that the escort phase progresses in the reverse sequence of $T$ when it reaches the initiator (in the case of a full cycle) the initiator initiates the cleaning phase by clearing $B_1$ (In the case of a sub-cycle the processor that detects the sub-cycle is the one that initiates the cleaning phase when its corresponding output buffer updates its state to the escort phase. For instance, in Figure \ref{Tokennn}, (c), the processor $q$ detects the cycle. $q$ initiates the cleaning phase when it updates the state of $OUT_q(r)$ to $(id,?,rqi,E,?)$). In the same manner $B_{k-1}$ will clear its state and so on. Thus in this case too we are sure that all the buffers part of $T$ will clear their state in a finite time. Observe that in the case $T$ found a free buffer the cleaning phase is initiated by the processor with the buffer $B_k$. $B_{k-1}$ clears then its state and so on. Thus we are sure that all the buffers part of $T$ will clear their state in a finite time and the lemma holds.
\end{proof}


\begin{lem}\label{Tokenp}
In the case where $Token_{p}(q)=true$, it will be set at false in a finite time.
\end{lem}

\begin{proof}
Note that in the case $Token_{p}(q)=true$ and the rule that allows the initiation of the token circulation is enabled, $Token_{p}(q)$ will be set at false by the token circulation algorithm when the this rule is executed. Otherwise, the two cases below are possible:
\begin{itemize}
\item $S_{pqi}=(id,?,?,?,?)$, in this case $Token_{p}(q)$ will be set at false by the Token Circulation algorithm by executing $R39$ (Note that $R39$ is enabled on $p$ and the daemon is weakly fair).
\item $S_{pqi}=(-1,NULL,NULL,C,-1)$. In the case the next processor by which the message that is in $IN_p(q)$ have to pass to reach the destination is $q$ then the rule that allows the initiation of the token circulation is enabled on $p$. Thus, $Token_{p}(q)$ will be set at false by the token circulation algorithm when the this rule is executed. Otherwise, $Token_{p}(q)$ will be set at false by executing $R'17$ that is enabled on $p$.
\end{itemize}

From the cases above we can deduct that in the case $Token_{p}(q)=true$, it will be set at false in a finite time and the lemma holds.
\end{proof}


\begin{lem}
If there is a processor that wants to generate a token circulation, it will be able to do it in a finite time.
\end{lem}

\begin{proof}
From Lemma \ref{Tokenp} we know that if $Token_p$ is true then it will be set at false in a finite time. From Lemma \ref{EXT} if $EXT_p$ is occupied, then it will be cleared in a finite time. From Lemma \ref{TokenClear} and Lemma \ref{AbnSequence} we know that if there is a token circulation that is executed all the buffers part of it will clear their state in a finite time. Thus when $p$ wants to generate a token circulation it will be able to do it in a finite time.
\end{proof}

\begin{lem}\label{TokenWin}
If there are some Token Circulations that are initiated then at least one of them will validate all its path.
\end{lem} 

\begin{proof}
Let us focus on the token circulation that has the smallest $id$ (Let this token be $T1$). When such a token circulation is initiated, the only things that can stop its progression is the presence on the path of another token circulation $T2$ that is in the Valid phase. Thus the following cases are then possible: 
\begin{enumerate}
 \item\label{C1} i) $T2$ is a correct token circulation. In this case two sub cases are possible as follow: i) all the path of $T2$ has been validated. No other token circulation can break $T2$. Thus according to Lemma \ref{TokenClear}, we are sure that the state of all the buffers of the path will be clean in a finite time. Thus $T1$ can continue its progression. ii) There is another token circulation $T3$ that cut $T2$. Note that in this case there is a part of the path that has been broken. An abnormal sequence is then created (Note that the buffers that were part of $T2$ that are in the valid phase are part of the abnormal sequence). According to Lemma \ref{AbnSequence}, the state of the buffers of the sequence will cleared. Thus $T1$ can continue its progression. 
 \item ii) $T2$ is not a correct token circulation. In this case $T2$ is an abnormal sequence. In this case according to Lemmas \ref{AbnSequence}. The state of the buffers part of $T2$ will be cleared in a finite time. Thus $T1$ can continue its progression. Note that $T2$ can behave as a valid token circulation. In this case we retrieve case \ref{C1}. 
 
 In both cases $T1$ continues its progression. Thus we are sure that $T1$ will be able to reach the last buffer $B_i$ such as $B_i$ is either empty or it wants to send the token to a buffer that is already in the path of $T1$. Note that on the processor that contains $B_i$ either $R4$ or $R5$ or $R6$ or $R7$ are enabled. The second phase is then initiated (the state of $B_i$ will be valid). It is easy to prove by induction that all the buffer on the path of $T1$ will be validated since that we are sure that there is no other token circulation that can break $T1$ (Recall that $T1$ has the smallest $id$). Thus the lemma holds. 
\end{enumerate}
\end{proof}

We can then deduce that at least one message will undergo a route change. The next lemma follows:


\begin{lem}\label{MSGSUITBUFF}
When the routing tables are stabilized all the messages will be in a suitable buffer in a finite time.
\end{lem}

\begin{proof}
Note that when the routing tables are stabilized, some messages may be on the wrong direction, however, we are sure that the number of such messages will never increase since both the generation and the routing of messages is always done in the right buffer (Recall that the routing tables are stabilized). In another hand according to Lemma \ref{EXT}, if the extra buffer of $p$ ($EXT_p$) is occupied, it will be free in a finite time. Suppose that $p$ is the processor that has an input buffer that contains a message $m$ that is not in a suitable buffer. This process will initiates a token circulation. According to Lemma \ref{TokenWin}. There is at least one token circulation that will finish its execution (Suppose that this token circulation is the one that was initiated by $p$). Thus we are sure that the output buffer of $p$ (the next destination of $m$) will be free in a finite time (refer to Lemma \ref{EndFreeBuffer}. Thus the message in $m$ will be copied in the free output buffer. Note that once it is copied in the corresponding output buffer, it becomes in a suitable buffer. Hence the number of the messages that are not in a suitable buffer decreases at each time. Thus we are sure that at the end all the messages will be in the right direction and hence in a suitable buffer and the lemma holds. 
\end{proof}


\begin{lem}\label{Token}
When the routing tables are stabilized and all the messages are in suitable buffer, no Token circulation is initiated.
\end{lem}

\begin{proof}
According to Lemma \ref{Tokenp}. For any $q\in N_p$ $Token_{p}(q)$ will be set at false in a finite time. Note that the only rule that set $Token_{pq}$ at true is $R'10$. However $R'10$ is never enabled since all the messages on the system are in suitable buffer and since the routing tables are correct (all messages are generated and routed in suitable buffers as well). Thus the lemma holds.
\end{proof}

The fair pointer mechanism cannot be disturbed anymore by the token circulations. Note that our buffer graph is a DAG when the routing tables are stabilized Thus:


\begin{lem}\label{Progresss}
All the messages progress in the system.
\end{lem}

\begin{proof}

In order to prove the lemma, it is sufficient to prove that all the buffers are continuously free. Note that if $\exists$ $q \in N_p$ such as , $IN_p(q)$ is free then if there is a message in $OUT_q(p)$, this message is automatically copied in $IN_p(q)$ and thus $OUT_q(p)$ becomes free. Hence it is sufficient to prove that the input buffers are free in a finite time. To do so, let's prove that $\forall$ $p \in I$, when there is a message in $IN_p(q)$, this message is deleted from $INp(q)$ in a finite time ($q \in N_p$). 

Recall that after the stabilization of the routing tables all the messages will be in the right direction and no token circulation is initiated (refer to Lemma \ref{MSGSUITBUFF} and \ref{Token}). Let consider the system after the stabilization of the routing tables and when all the messages are in the right direction. 
Let consider the message $m$ that is in the input buffer of the processor $p$, referred to as $B_1$. Let $B_1, B_2, B_3, ..., B_k$ the be the sequence of buffers starting from $B_1$such as $B_i=m'$ and $B_{i+1}$ is the next buffer by which $m'$ should pass by to reach its destination. Note that $B_i$ is an input buffer when $i$ is odd. In the worst case $\forall$ $1<i\leq k$ all the buffers are full and $B_k$ is the input buffer of a leaf processor that we will call $p_0$ (Recall the all the messages are on the right direction). Note that the input-buffers in the sequence $B_1, B_2, B_3, ..., B_k$ are all at an even distance from the input buffer $B_k$. Let define $\delta$ as the distance between the input buffer of the processor $p$ and the input buffer of processor $p_0$ ($B_k$). The lemma is proved by induction on $\delta$. We define for this purpose the following predicate $P_\delta$:
If there is a message m in $B_i$ such as $B_i$ is an input buffer and at distance $\delta$ from the input buffer $B_k$ then one of these two cases happens:
\begin{itemize}
\item $m$ is consumed and hence delivered to its destination.
\item $m$ is deleted from the input buffer and copied in $B_{i+1}$ (which is an output buffer).
\end{itemize}

\paragraph{\textbf{Initialization.}} Let's prove the result for $B_k$ on $p_0$. Suppose that there is a message $m$ in $B_k$. In this case we are sure that $p_0$ is the destination of the message (otherwise the message $m$ is in the wrong direction since $p_0$ is a leaf processor). Thus, in this case, since the daemon is weakly fair and since $R'2$ keep being enabled on $p_0$ then $R'2$ will be executed in a finite time and the message $m$ in $B_k$ is consumed. Thus $P_0$ is true.

\paragraph{\textbf{Induction.}} let $\delta$ ≥ 1. We assume that $P_{2\delta}$ is true and we prove that $P_{2\delta+2}$ is true as well (Recall that the input buffers are at an even distance from the input buffer $B_k$ that in the processor $p_0$). Let $B_i$ on be the input buffer of $p$ that is at distance $2\delta$ from $B_k$ and $B_{i-2}$ the one that is on $p'$ being at distance $2\delta + 2$ from $B_k$ containing the message $m′$. In the case where the destination of $m′$ is $p'$ then it will be consumed in a finite time (the daemon is weakly fair and $R'2$ keep being enabled on $p'$. Thus $p'$ will execute $R'2$ in a finite time). Hence $P_{2\delta+2}$ is in this case true. In the other case (the destination of $m′$ is different from $p'$), since $P_{2\delta}$ is true then if there is a message $m$ in $B_i$ then we are sure that this message will be either consumed or copied in $B_{i+1}$. Thus $B_i$ becomes a free buffer. The cases bellow are possible according to the rule that is executed on $B_{i-1}$:
\begin{enumerate}
\item \label{Progress} $R'3$ is executed. In this case one message that is in an input buffer of $p$ is copied in $B_{i-1}$. However, since the pointer on $B_{i-1}$ is fair, we are sure that the message $m$ in $B_{i-2}$ will be copied in $B_{i-1}$. Thus $B_{i-2}$ will be free in a finite time and the lemma holds.
\item a message $m'$ is generated in $B_{i-1}$. However since $P_{2\delta}$ is true $B_i$ becomes free in a finite time thus $m'$ will be copied in $B_i$ in a finite time. Thus $B_{i-1}$ becomes free. Nevertheless, since one message has been generated in the previous step, we are sure that $R'3$ will be the one that will be executed. Thus we retrieve Case \ref{Progress}.
\end{enumerate}
\end{proof}

\begin{lem}\label{MessageGeneration}
Any message can be generated in a finite time under a weakly fair daemon.
\end{lem}  

\begin{proof}
According to Lemma \ref{Token}, no token is initiated when the routing tables are stabilized and when all the messages are in suitable buffers, thus the fair pointer mechanism cannot be disturbed by the token circulation anymore. Note that since the routing tables are stabilized and since the buffer graph is a DAG no deadlocks happens. Thus all the messages progress in the system. Suppose that the processor $p$ wants to generate a message. Recall that the generation of a message $m$ for the destination $d$ is always done in the output buffer of the processor $p$ connected to the link $(p, q)$ such as $Next_p(d) = q$. Two cases are possible:
\begin{enumerate}
\item \label{YESF} $OUT_p(q) = \epsilon$ . In this case, the processor executes either $R'1$ or $R'3$ in a finite time. the result of this execution depends on the value of the pointer. Two cases are possible:
\begin{itemize}
\item the pointer refers to $R'1$. Then $p$ executes $R'1$ and hence it generates a message. Thus we obtain the result.
\item the pointer refers to $R'3$. Then $p$ executes $R'3$ in a finite time. Hence $OUT_p(q)\ne \epsilon$ and we retrieve case \ref{NOTF}. Note that the fairness of the pointer guarantees us that this case cannot appear infinitely.
\end{itemize}
\item \label{NOTF} $OUT_p(q)\ne \epsilon$. Since all the messages move gradually in the buffer graph we are sure that $OUT_p(q)$ will be free in a finite time and we retrieve \ref{YESF}.
\end{enumerate}
We can deduct that every processor can generate a message in a finite time. 
\end{proof}


We can now state the following Theorem:

\begin{theorem}\label{NoDeadlockStarvation}
Neither deadlock nor starvation situations appear in the system.
\end{theorem}

\begin{proof}
According to Lemma \ref{Progresss}. All the messages progress in the system. Thus we are sure that there is no message that stays locked on one buffer. in another hand according to Lemma \ref{MessageGeneration}, every processors will be able to generate a message. Hence the Theorem holds.
\end{proof}


\begin{lem}\label{NoMsgDuplicated}
The forwarding protocol never duplicates a valid message even if the routing algorithm runs simultaneously. 
\end{lem}

\begin{proof}
Let consider the message $m$. The cases below are possible:
\begin{itemize}
\item $m$ is in $EXT_p$. $m$ is then either deleted or copied in $OUT_p(q)$. Since this operation is a local operation (the copy is done between two buffer of the same processor) $m$ is copied in the new buffer and deleted from the previous one in a sequential manner.
\item $m$ is in $IN_p(q)$. The cases are then possible:
          \begin{itemize}
             \item $m$ is consumed ($R'2$ is executed). The message $m$ is deleted since a new value overwrites it.
             \item $m$ is copied in the extra buffer ($R'11$ or $R'12$ is executed). The message $m$ is copied in the extra buffer and deleted from the input buffer since in both cases a new value overwrites it.
             \item $m$ is copied in the output buffer ($R'3$ or $R'4$ is executed). Note that this operation is a local operation. Thus $m$ is copied in the output buffer an deleted from the input buffer (a new value overwrites it).
          \end{itemize}
\item $m$ is in $OUT_p(q)$. $m$ is then copied in the input buffer of the processor $q$ ($IN_q(p)$). Hence two copies are in the system. However $m$ in $IN_q(p)$ is neither consumed nor transmitted unless the copy in $OUT_p(q)$ is deleted  (see Rules $R'2$, $R'3$ and $R'4$).
\end{itemize}

From the cases above we can deduce that no message is duplicated in the system.
\end{proof}


\begin{theorem}
The proposed algorithm (Algorithms 1, 2 and 3) is a snap-stabilizing message forwarding algorithm (satisfying $SP$) under a weakly fair daemon.
\end{theorem}

\begin{proof}
From Theorem \ref{NoMsgDeleted} no valid message is deleted. From Theorem \ref{NoDeadlockStarvation} There is no deadlocks in the system and all the processors are able to generates messages in a finite time. From Lemma \ref{NoMsgDuplicated}, no message is duplicated. Hence, the theorem holds.

\end{proof}

  
\section{Conclusion}\label{sec:Conclusion}    
In this paper, we presented the first snap-stabilizing message forwarding protocol on trees that uses a number of
buffers per node being independent of any global parameter. Our protocol uses only $4$ buffers per link and an
extra one per node. This is a preliminary version to get a solution that tolerates topology changes provided
that the topology remains a tree.

\newpage

\begin{scriptsize}
\bibliographystyle{splncs}
\bibliography{Token}
\end{scriptsize}

\end{document}